\documentclass[aps,prx,twocolumn]{revtex4-1}

\usepackage{amsmath, amsthm, amssymb}
\usepackage{graphicx,color,psfrag}
\usepackage[normalem]{ulem}
\usepackage{hyperref}
\usepackage{cleveref}
\usepackage{caption}
\usepackage{subcaption}
\usepackage{thmtools,mathtools}
\declaretheoremstyle[notefont=\bfseries,notebraces={}{},%
    headpunct={},postheadspace=1em]{mystyle}
\declaretheorem[style=mystyle,numbered=yes,name=Lemma]{LemmaSecondLawNumberCurrent}
\declaretheorem[style=mystyle,numbered=yes,name=Postulate]{PostulatePhysicalProbes}
\declaretheorem[style=mystyle,numbered=yes,name=Lemma,sibling=LemmaSecondLawNumberCurrent]{LemmaSecondLawHeatCurrent}
\declaretheorem[style=mystyle,numbered=yes,name=Theorem]{TheoremOnsager}
\declaretheorem[style=mystyle,numbered=yes,name=Theorem,sibling=TheoremOnsager]{TheoremUniqueness}
\declaretheorem[style=mystyle,numbered=yes,name=Theorem,sibling=TheoremUniqueness]{TheoremExistence}
\declaretheorem[style=mystyle,numbered=yes,name=Corollary,parent=TheoremUniqueness]{CorollaryExistence}
\DeclarePairedDelimiter\abs{\lvert}{\rvert}%
\DeclarePairedDelimiter\norm{\lVert}{\rVert}%
\makeatletter
\let\oldabs\abs
\def\abs{\@ifstar{\oldabs}{\oldabs*}}
\let\oldnorm\norm
\def\norm{\@ifstar{\oldnorm}{\oldnorm*}}
\makeatother

\usepackage{color}

\newcommand{\Lfun}[2]{{\cal L}^{(#1)}_{#2}}
\newcommand{\Tr}[1]{\operatorname{Tr} \left\{ #1 \right\}}

\DeclareMathAlphabet{\gcal}{OMS}{cmsy}{m}{n}
\newcommand{\myT}{{\gcal T}}

\graphicspath{{./Figs/}}

\begin{document}

\title{Temperature and voltage measurement in quantum systems far from equilibrium}
\author{Abhay Shastry}
\affiliation{Department of Physics, University of Arizona, 1118 East Fourth Street, Tucson, AZ 85721}
\author{Charles\ A.\ Stafford}
\affiliation{Department of Physics, University of Arizona, 1118 East Fourth Street, Tucson, AZ 85721}
\date{\today}
\begin{abstract}
We show that a local measurement of temperature and voltage for a quantum system
in steady state, arbitrarily far from equilibrium, with arbitrary interactions within the system, is unique when it exists.
This is interpreted as a consequence of the second law of thermodynamics. We further derive a necessary and sufficient condition
for the existence of a solution. In this regard, we find that a positive temperature solution exists whenever there is no net population inversion.
However, when there is a net population inversion, we may characterize the system with a unique negative temperature. Voltage and temperature
measurements are treated on an equal footing:
They are simultaneously measured in a noninvasive manner, via a weakly-coupled {\em thermoelectric probe}, 
defined by requiring vanishing charge {\em and} heat dissipation into the probe.
Our results strongly suggest that a local temperature measurement without a simultaneous local voltage measurement, or vice-versa, is a misleading characterization of
the state of a nonequilibrium quantum electron system. These results provide a firm mathematical foundation for voltage and temperature
measurements far from equilibrium.
\end{abstract}
\maketitle

\section{Introduction}
\label{sec:Introduction}

Scanning probe microsopy \cite{Binnig1982a,Binnig1983b,Binnig1986,Chen93,Kalinin2007} revolutionized the field of nanoscience 
and enabled the measurement of local thermodynamic observables such as voltage \cite{Muralt1986} and temperature \cite{Williams1986}
in nonequilibrium quantum systems.
The ability to define local thermodynamic variables in a system far from equilibrium 
is of fundamental interest because it is a necessary step toward the construction of
nonequilibrium thermodynamics \cite{Ruelle2000,Casas-Vazques2003,Lebon2008,Cugliandolo2011,Jacquet2012,Stafford2014,Esposito2015,Shastry2015}.
Many experiments in mesoscopic electrical transport utilize voltage probes as circuit elements \cite{Picciotto2001,Gao2005,Benoit1986,Shepard1992},
and scanning potentiometers are now a mature technology \cite{Kanagawa2003,Bannani2008,Luepke2015}, routinely achieving sub-angstrom spatial resolution
to study a host of novel physical phenomena \cite{Wang2013,Clark2014,Willke2015,Yamasue2015}.
In contrast, scanning thermometry \cite{Williams1986} has proven significantly more challenging \cite{Majumdar1999}, but is
currently undergoing a rapid evolution toward nanometer resolution \cite{Kim2011,Yu2011,Kim2012,Menges2012},
leading to important insights into transport and dissipation at the nanoscale \cite{Agrait2002,Ward2011,Lee2013,Kim2015,Hu2015}.
A fundamental challenge for theory is to develop a rigorous mathematical description of such local thermodynamic measurements.
Until now, mainly operational definitions 
\cite{Engquist81,Dubi2009b,Jacquet2009,Dubi2009,Caso2011,Jacquet2012,Sanchez2011,Caso10,Bergfield2013demon,Bergfield2015,Ye2015}
have been advanced, leading to a competing panoply of often contradictory
definitions of such basic observables as temperature and voltage.

The second law of thermodynamics is one of the cornerstones of physics.
The origin of the second law 
was rooted in empirical observations in the early nineteeth century,
and its theoretical explanation emerged with the gradual development of the statistical foundation of thermodynamics.
The statistical basis of the second law 
places it in a league of its own, among the laws of physics.
A quote on the subject, at once exalting and to the point, by the famous astrophysicist Sir Arthur Eddington reads as follows \cite{Eddington1929}:
``The law that entropy always increases holds, I think, the supreme position among the laws of Nature. 
If someone points out to you that your pet theory of the universe is in disagreement with Maxwell's equations --- 
then so much the worse for Maxwell's equations. 
If it is found to be contradicted by observation --- well, these experimentalists do bungle things sometimes. 
But if your theory is found to be against the second law of thermodynamics I can give you no hope; 
there is nothing for it but to collapse in deepest humiliation." Any theory which purports to describe the measurement of temperature, voltage or other
thermodynamic parameters, must therefore satisfy this fundamental requirement, and as Eddington notes, regardless of the nature of microscopic interactions.

We examine statements of the second law of thermodynamics, accompanied with mathematical proofs,
and their consequences, in the context of local noninvasive measurements of temperature and voltage in nonequilibrium quantum electron systems.
We derive our results from very general considerations, i.e.,
for electron transport in steady state, arbitrarily far from equilibrium, and for arbitrary interactions within 
the quantum electron system. 
Our considerations apply to any system of fermions, charged or neutral.
While our analyses in this article are presented in a
theorem-proof format, their motivation draws from physical
principles. We show that the uniqueness of the temperature and voltage measurement
is a consequence of the second law of thermodynamics and that, in order to obtain a unique measurement, it is necessary to measure both
temperature and voltage simultaneously. Simply put, this is because electrons carry both energy and charge.

In order to have a meaningful definition of temperature, the Hamiltonian must be bounded below ($\langle H\rangle\geq-c$ for some finite $c\in\mathbb{R}$).
By the same token,
a system can, in principle, exhibit negative temperatures if the energy averaged over the spectrum is bounded above ($\langle H\rangle\leq c$ for some finite $c\in\mathbb{R}$). These are well-known results in statistical physics, and we highlight their
role in the context of local noninvasive measurements of temperature and voltage.
We derive a condition, that is both necessary and sufficient, for the existence of a 
joint temperature and voltage measurement.
This condition corresponds physically to a nonequilibrium system that does not exhibit local population inversion. We obtain also, as
a corollary of the former condition, the result that nonequilibrium systems exhibiting local population inversion can be characterized with a 
negative temperature which is also unique. Population inversion is the working principle behind important Fermionic devices such as the maser and laser
\cite{Einstein1917,Kastler1957,Schawlow1958,Scully1967}.

In this article, we consider a probe that couples exclusively to the electronic degrees of freedom.  
Out of equilibrium, the temperature distributions of different microscopic
degrees of freedom (e.g., electrons, phonons, nuclear spins) do not, in general, coincide, so that one has to distinguish between measurements of the electron
temperature \cite{Engquist81,Dubi2009b,Bergfield2013demon,Meair14} and the lattice temperature \cite{Ming10,Galperin2007}.  
This distinction is particularly acute in the extreme limit of
elastic quantum transport \cite{Bergfield09}, where electron and phonon temperatures are completely decoupled.
It should be emphasized that the electrons within the system are free to undergo arbitrary interactions, e.g., with photons, phonons, other electrons, etc.  
However, {\em direct} heat transport into the probe via black-body
radiation, phonons, etc.\ is excluded.  Inclusion of these additional heat transfer processes into the probe leads to a temperature
measurement that is simply a combination of the temperatures of the various microscopic degrees of freedom \cite{Bergfield2015}.

The article is organised as follows. We outline the formalism in Sec.\ \ref{sec:Formalism}, and
introduce a postulate that helps put our results on sound mathematical footing. In Sec.\ \ref{sec:LocalMeasurements}, we discuss
our theory of local thermodynamic measurements, explain the idea behind noninvasive measurements, and also derive some useful expressions for further
analysis. In Sec.\ \ref{sec:Uniqueness}, we provide several statements of the second law of thermodynamics and show their relation to the uniqueness
of temperature and voltage measurements.
In Sec.\ \ref{sec:existence}, we start by defining certain useful quantities and proceed to derive the condition for the existence of a solution.
We also discuss here the case of broadband probes in order to further illustrate the physical meaning behind our results, and conclude that
probes operating in the broadband limit can be considered to be {\em ideal}.
We consider the other extreme as well, i.e., narrowband probes and conclude that they are unsuitable for measurements. Our results
are illustrated for a two-level system which is detailed in Sec.\ \ref{sec:TwoLevelSystem}.
We conclude with a summary of our central findings in Sec.\ \ref{sec:conclusions}, contrasting our approach to prior theoretical work, and discuss possible
future directions.
Some key results on the local properties of fermions in a nonequilibrium steady state are presented in Appendix \ref{sec:AppA}, which are needed in
our analysis of the measurement problem.

\section{Formalism}
\label{sec:Formalism}

We use the nonequilibrium Green's function formalism (NEGF) for describing the motion of electrons within a quantum conductor.
A general expression for the nonequilibrium steady-state electrical current ($\nu=0$) \cite{Meir92} and the electronic contribution to the heat current ($\nu=1$) \cite{Bergfield09b}
flowing into a macroscopic electron reservoir $P$,
can be written in a form analogous to the two-terminal
Landauer-B{\"u}ttiker formula \cite{Stafford2014}:
\begin{equation}
\begin{aligned}
I_{p}^{(\nu)}=\frac{1}{h}\int_{-\infty}^{\infty}d\omega(\omega-\mu_{p})^{\nu}&\myT_{ps}(\omega)[f_{s}(\omega)-f_{p}(\omega)],\\
\text{with}\ \nu &=\{0,1\},
\label{Rearraged}
\end{aligned}
\end{equation}
where one may think of
\begin{equation}
\myT_{ps}(\omega)=2\pi \Tr{\Gamma^{p}(\omega)A(\omega)}
\label{tps}  
\end{equation}
as a local transmission function between the
macroscopic probe terminal and the nonequilibrium quantum system. $f_{s}(\omega)$ is the nonequilibrium distribution function
of the system, as sampled by the probe, and is defined by \cite{Stafford2014}
\begin{equation}
f_{s}(\omega)\equiv\frac{\Tr{\Gamma^{p}(\omega)G^{<}(\omega)}}{2{\pi}i\Tr{\Gamma^{p}(\omega)A(\omega)}}.
\label{nonequilibriumdistribution}
\end{equation}
In Eqs.\ (\ref{tps}--\ref{nonequilibriumdistribution}), $A(\omega)=\big(G^{<}(\omega)-G^{>}(\omega)\big)/2\pi i$
is the spectral function. $G^{<}(\omega)$ and $G^{>}(\omega)$ are the Fourier transforms of the 
Keldysh ``lesser'' and ``greater'' 
Green's functions \cite{Stefanucci2013}, 
describing the nonequilibrium electron and hole distributions within the system, respectively (see Appendix \ref{sec:AppA} for details). 
$\Gamma^{p}(\omega)$
is the tunneling width matrix describing the coupling of the probe to the system, and $f_{p}(\omega;\mu_{p},T_{p})=1/\big(1+\exp(\frac{\omega-\mu_{p}}{T_{p}})\big)$
is the Fermi-Dirac distribution of the probe. We note that the expression in
Eq.\ (\ref{Rearraged}) is completely general and allows for arbitrary interactions within the quantum system, and arbitrary bias conditions
of the reservoirs.

Since the spectral function $A(\omega)$ is positive-semidefinite and the probe-system coupling $\Gamma^{p}(\omega)$  
is positive-definite (see Appendix \ref{sec:AppA}), we note
that 
\begin{equation}
\begin{aligned}
\Tr{A(\omega)\Gamma(\omega)}&=\Tr{A(\omega)^{1/2}A^{1/2}(\omega)\Gamma(\omega)}\\
&=\Tr{A^{1/2}(\omega)\Gamma(\omega) A^{1/2}(\omega)}\\ &\geq0,
\end{aligned}
\end{equation}
where $A^{1/2}(\omega)$ is the positive-semidefinite square root of $A(\omega)$. $A^{1/2}(\omega)\Gamma(\omega) A^{1/2}(\omega)$ becomes
positive-semidefinite when $A^{1/2}(\omega)$ and $\Gamma(\omega)$ are positive-semidefinite \cite{Horn1986} and therefore we have
\begin{equation}
\myT_{ps}(\omega)\geq0,\ \forall \omega\in \mathbb{R}.
\label{tps_positive}
\end{equation}

We note that $f_{s}(\omega)$ satisfies the property of a distribution function, namely,
\begin{equation}
0\leq f_{s}(\omega)\leq1\ \ \forall \omega\in \mathbb{R},
\label{fsBounded}
\end{equation}
as shown in appendix \ref{sec:AppA}. 
We start our analysis with the following postulate, and explain its physical significance.
\begin{PostulatePhysicalProbes}
\label{postulate}
The local probe-system transmission function $\myT_{ps}: \mathbb{R}\to[0,\infty)$ and the nonequilibrium distribution function $f_{s}:\mathbb{R}\to[0,1]$
are measurable over any interval $[a,b]\in\mathbb{R}$, and
$\myT_{ps}(\omega)$ satisfies
\begin{equation}
\label{finiteParticleRate}
0<\int_{-\infty}^{\infty}d\omega\myT_{ps}(\omega) <\infty,
\end{equation}
and
\begin{equation}
\label{finiteEnergyRate}
\bigg|\int_{-\infty}^{\infty}d\omega\ \omega\myT_{ps}(\omega)\bigg|<\infty.
\end{equation}
\end{PostulatePhysicalProbes}

The measurability of $\myT_{ps}(\omega)$ and $f_{s}(\omega)$ is taken to lend meaning to the currents  in Eq.\ (\ref{Rearraged}). 
We point out that the finiteness of the two integrals given in
Eqs.\ (\ref{finiteParticleRate}) and (\ref{finiteEnergyRate}) is more relevant to our discussion of 
existence in Sec. \ref{sec:existence}.
Our result on uniqueness, as stated in Theorem \ref{ThmUniqueness}, 
is somewhat stronger and requires only that the function $\myT_{ps}(\omega)$
grow slower than exponentially for large values of energy (for $\omega\rightarrow\pm\infty$). 

On physical grounds, the probe-sample transmission function $\myT_{ps}(\omega)$ can be argued to have a compact support (non-zero only
for some finite interval $[\omega_{-},\omega_{+}]\subset\mathbb{R}$). It is easy to see that $\myT_{ps}$ must have a lower bound $\omega_{-}$  such
that $\myT_{ps}(\omega)=0\ \forall\ \omega<\omega_{-}$, since physical Hamiltonians must have a finite ground-state energy.
However, for energies larger than the probe work function ($\omega_{+}$), it can be argued that the particle will merely pass through 
the probe and
not contribute to the steady state currents into the probe. $\myT_{ps}(\omega)$ then has a compact support and satisfies Eqs.\ 
(\ref{finiteParticleRate}) and (\ref{finiteEnergyRate}).
In section \ref{sec:existence}, we comment upon the limiting case where the measure of $\omega\myT_{ps}(\omega)$ in 
Eq.\ (\ref{finiteEnergyRate}) tends to infinity. The absolute value on the $lhs$ in Eq.\ (\ref{finiteEnergyRate})
is somewhat redundant since the
limiting case must have $\omega_{+}\rightarrow\infty$
while $\omega_{-}\rightarrow-\infty$ is ruled out based on the principle that any physical spectrum has a finite ground-state energy.
We note that Eqs.\ (\ref{finiteParticleRate}), (\ref{finiteEnergyRate}) also imply
\begin{equation}
0<\int_{-\infty}^{\infty}d\omega\myT_{ps}(\omega)f_{s}(\omega),\int_{-\infty}^{\infty}d\omega\myT_{ps}(\omega)f_{p}(\omega) <\infty
\label{BoundedParticleRates}
\end{equation}
and
\begin{equation}
\int_{-\infty}^{\infty}d\omega\ \omega\myT_{ps}(\omega)f_{s}(\omega),\int_{-\infty}^{\infty}d\omega\ \omega\myT_{ps}(\omega)f_{p}(\omega) <\infty.
\label{BoundedEnergyRates}
\end{equation}

\section{Local Measurements}
\label{sec:LocalMeasurements}

The local voltage and temperature of a nonequilibrium quantum system, as measured by a scanning
thermoelectric probe, is defined by
the simultaneous conditions of vanishing net charge dissipation {\em{and}} vanishing net heat dissipation into the
probe \cite{Bergfield2013demon,Meair14,Bergfield2014,Bergfield2015,Stafford2014,Shastry2015}:
\begin{equation}
I_{p}^{(\nu)}=0,\ \ \ \ \nu \in \{0,1\},
\label{equilibrium}
\end{equation}
where $\nu=0,1$ correspond to the electron number current and the electronic contribution to the heat current, 
respectively. Eq.\ (\ref{equilibrium}) gives the conditions under
which the probe is in local equilibrium with the sample, which is itself arbitrarily far from equilibrium. 

We define the system's local temperature
and voltage using a probe that is weakly coupled via a tunnel barrier. The other end of this scanning probe \cite{Chen93,Kalinin2007} is the macroscopic
electron reservoir whose temperature and voltage are both adjusted until Eq.\ (\ref{equilibrium}) is satisfied.
A weakly coupled probe is a useful theoretical construction for our analysis, and the extension of our
results beyond the weak-coupling limit is an open question. We explain the physical basis of weak coupling below, 
and derive some useful formulae.

\begin{figure}
\centering
\captionsetup{justification=raggedright,
singlelinecheck=false
}
\includegraphics[width=3.4in]{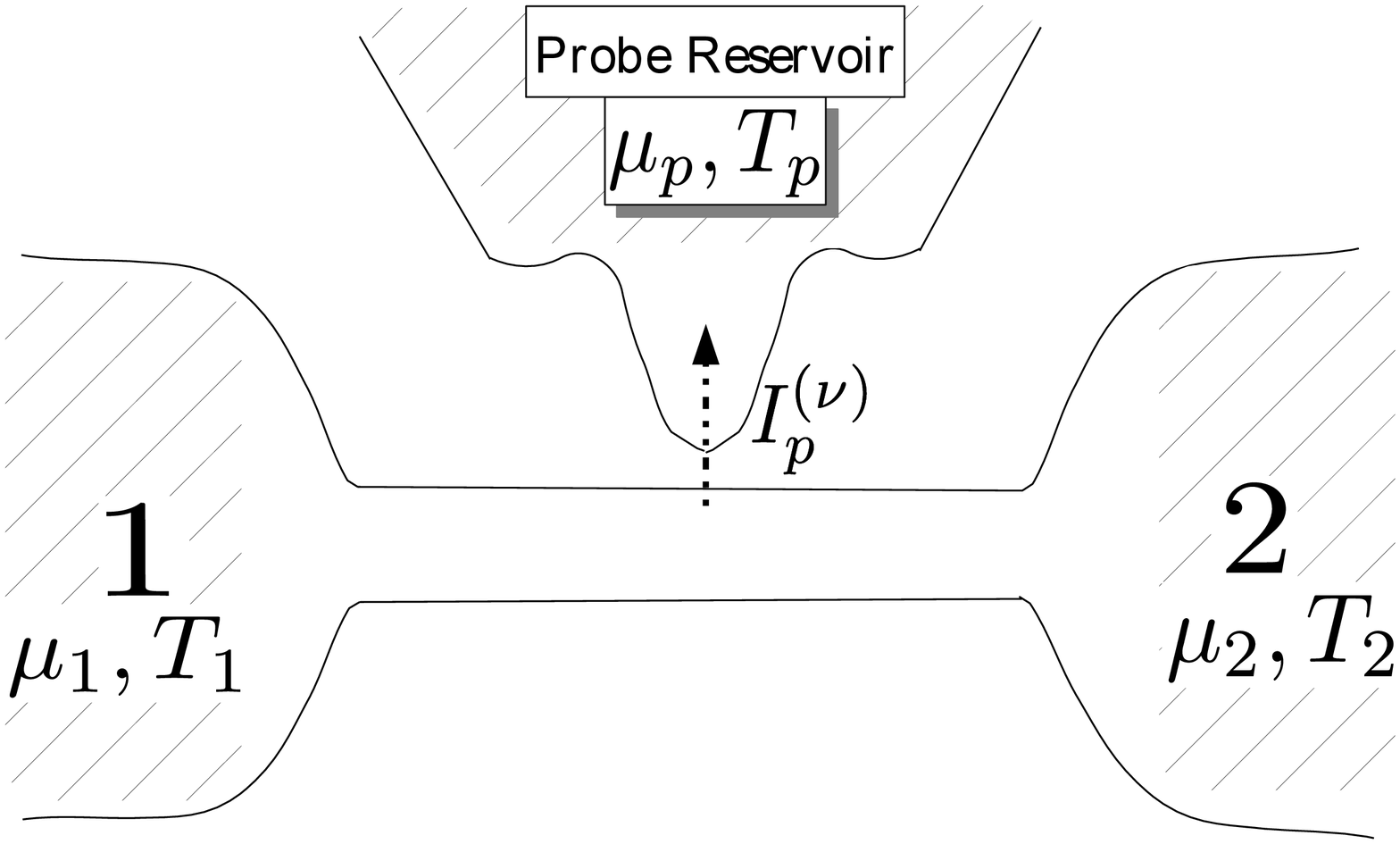}
\caption{\label{fig:Setup}Illustration of the measurement setup: The quantum conductor represented below is
in a nonequilibrium steady state. A weakly-coupled scanning tunneling probe
noninvasively measures the local voltage ($\mu_{p}$) and local temperature ($T_{p}$) {\em simultaneously}: By requiring
both a vanishing net charge exchange ($I^{(0)}_{p}=0$){\em and} a vanishing net heat exchange ($I^{(1)}_{p}=0$) with the system.
The nonequilibrium steady state has been prepared, in this particular illustration, via the electrical and thermal 
bias of the strongly-coupled reservoirs ($1$ and $2$). The measurement method itself is completely general
and does not depend upon (a) how such a nonequilibrium steady-state is prepared, (b) how far from equilibrium the quantum electron system is
driven, and (c) the nature of interactions within that system.
}
\end{figure}

\subsection{Noninvasive measurements}

When the coupling of the probe to the system is weak, we may take $\myT_{ps}(\omega)$ in Eq.\ (\ref{tps}) and
the local nonequilibrium distribution function $f_{s}(\omega)$ to
be independent of the probe Fermi-Dirac distribution $f_{p}(\omega)$. 
While both these quantities depend upon the local probe-system coupling in an obvious manner,
the weak-coupling condition essentially implies that the nonequilibrium steady state of the system is unperturbed by the introduction of the probe
terminal. The voltage and temperature of the probe itself play no role in preparing the nonequilibrium steady state. In other words,
the probe does {\em not} drive the system but merely exchanges energy and particles across a weakly-coupled tunnel barrier
and constitutes a {\em noninvasive} measurement.

Given a system prepared in a certain nonequilibrium steady state (e.g., by a particular bias of the strongly coupled reservoirs), the currents given
by Eq.\ (\ref{Rearraged}) are functions of the probe Fermi-Dirac distribution specified by its temperature and
chemical potential
\begin{equation}
I^{(\nu)}_{p}\equiv I^{(\nu)}_{p}(\mu_{p},T_{p}).
\end{equation}

It can be seen that the currents are continuous functions of $\mu_{p} \in (-\infty,\infty)$
and $T_{p} \in (0,\infty)$ with continuous gradient
vector fields defined by
\begin{equation}
\nabla{I^{(\nu)}_{p}}\equiv\bigg(\frac{\partial{I^{(\nu)}_{p}}}{\partial{\mu_{p}}},\frac{\partial{I^{(\nu)}_{p}}}{\partial{T_{p}}}\bigg).
\label{CurrentGradient}
\end{equation}
With $k_{B}$ set to unity, we compute the gradients of the currents using Eq.\ (\ref{Rearraged}).
We find the gradient of the number current to be
\begin{equation}
\label{NumberCurrentGradient}
\nabla{I^{(0)}_{p}}=\bigg(-\Lfun{0}{ps},-\frac{\Lfun{1}{ps}}{T_{p}}\bigg).
\end{equation}
The gradient of the heat current reduces to
\begin{equation}
\label{HeatCurrentGradient}
\nabla{I^{(1)}_{p}}=\bigg(-\Lfun{1}{ps}-I^{(0)}_{p},-\frac{\Lfun{2}{ps}}{T_{p}}\bigg),
\end{equation}
where we define the response coefficients $\Lfun{\nu}{ps}$ 
as
\begin{equation}
\begin{aligned}
\Lfun{\nu}{ps}&\equiv\Lfun{\nu}{ps}(\mu_{p},T_{p})\\&=\frac{1}{h}\int_{-\infty}^{\infty}d\omega(\omega-\mu_{p})^{\nu}\myT_{ps}(\omega)\bigg(-\frac{\partial{f_{p}}}{\partial{\omega}}\bigg),
\end{aligned}
\label{eq:Onsager}
\end{equation}
which are easily seen to be finite
\footnote{$\Lfun{\nu}{ps}(\mu_{p},T_{p})$ are finite even if $\myT_{ps}(\omega)$ and $\omega\myT_{ps}(\omega)$ do not
obey the finite measure conditions of postulate \ref{postulate} due to the exponentially decaying tails of the Fermi-derivative.
We merely need $\myT_{ps}(\omega)$ to grow slower than exponentially for $\omega\to\pm\infty$.}.

Although the coefficients $\Lfun{\nu}{ps}$ formally resemble the Onsager linear-response coefficients \cite{Onsager1931}
of an elastic quantum conductor \cite{Sivan1986},
it is very important to note that we do {\em not} make the assumptions of linear response, local equilibrium, or elastic transport
in the above definition of $\Lfun{\nu}{ps}$: The system itself may be arbitrarily far from equilibrium with arbitrary inelastic scattering processes.
The coefficients above appear
naturally when we calculate the gradient fields defined by Eq.\ (\ref{CurrentGradient}), and the gradient operator is of course
given by the first derivatives. Our main results follow from an analysis of the properties of these gradient fields.

\section{Uniqueness and the second law}
\label{sec:Uniqueness}
We now turn to one of the central problems which we set out to address:
$I_{p}^{(\nu)}(\mu_{p},T_{p})=0$, with $\nu=\{0,1\}$,
is a system of coupled nonlinear equations in two variables that defines our local voltage and temperature measurement.
There is no {\it a priori} reason to expect a unique solution, if a solution exists at all.
We begin the section with statements of the second law of thermodynamics,
and conclude by showing that the uniqueness of the measurement emerges as a consequence.

\subsection{Statements of the second law}
We note that $\forall\ \mu_{p}\in(-\infty,\infty)$ and $T_{p}\in(0,\infty)$,
\begin{equation}
\begin{aligned}
\Lfun{0}{ps}(\mu_{p},T_{p})&>0\\
\Lfun{2}{ps}(\mu_{p},T_{p})&>0,
\end{aligned}
\end{equation}
since $\myT_{ps}(\omega)\geq0$, and the measure of $\myT_{ps}(\omega)$ and the Fermi-function derivative are both nonzero and strictly positive. 
This leads to two statements of the second law of thermodynamics, related to the Clausius statement, which are
presented in the following two lemmas. The idea is to choose the correct contour for each case, and to evaluate the line integral over
the current gradients in Eqs.\ (\ref{NumberCurrentGradient}) and (\ref{HeatCurrentGradient}). A cursory glance at the
number current gradient in Eq.\ (\ref{NumberCurrentGradient}) suggests that the contour should be defined over a constant temperature, 
while the heat current gradient in Eq.\ (\ref{HeatCurrentGradient}) suggests a line integral over a constant voltage contour.

\begin{figure}
\captionsetup{justification=raggedright,
singlelinecheck=false
}\includegraphics[width=3.4in]{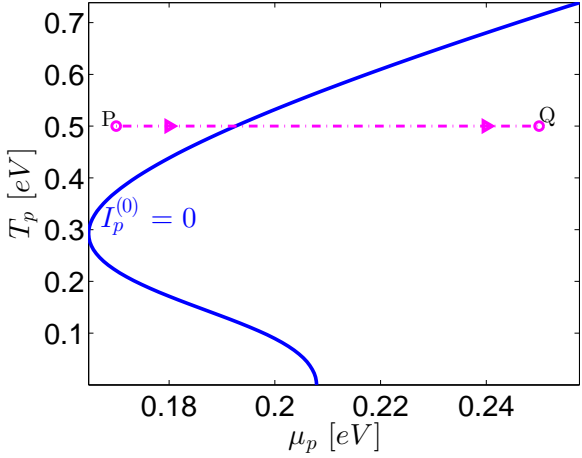}
\caption{\label{fig:Lemma1}Illustration of Lemma \ref{LemmaSecondLawNumberCurrent1}: The contour $PQ$ shown in magenta cuts the number
current contour $I^{(0)}_{p}=0$ (or any $I^{(0)}_{p}=\text{constant}$) exactly once.
The contour line from $P$ to $Q$ is at a constant temperature ($T_{p}=\text{constant}$), and illustrates the Clausius statement: 
The number current is monotonically decreasing along $PQ$.
The system and bias conditions are detailed in Sec.\ \ref{sec:TwoLevelSystem}.
}
\end{figure}

\begin{LemmaSecondLawNumberCurrent}
\label{LemmaSecondLawNumberCurrent1}
The number current contour defined by $I^{(0)}_{p}(\mu_{p},T_{p})=0$ exists for all $T_{p}\in (0,\infty)$ and defines a function $M:(0,\infty)\rightarrow\mathbb{R}$
where $\mu_{p}=M(T_{p})$, such that the second law of thermodynamics is obeyed:
\begin{equation}
\begin{aligned}
\label{lemma1}
I^{(0)}_{p}(\mu'_{p},T_{p})&>0,\text{if}\ \mu'_{p}<M(T_{p})\ \text{and}\\
I^{(0)}_{p}(\mu'_{p},T_{p})&<0,\text{if}\ \mu'_{p}>M(T_{p}).
\end{aligned}
\end{equation}
\begin{proof}
We first show that $I^{(0)}(\mu_{p},T_{p})=0$ is satisfied for all $T_{p}\in (0,\infty)$. For any $T_{p}\in (0,\infty)$, we have
\begin{equation}
\begin{aligned}
\lim_{\mu_{p}\to-\infty}I^{(0)}(\mu_{p},T_{p}) =& \frac{1}{h}\int_{-\infty}^{\infty}d\omega\myT_{ps}(\omega)[f_{s}(\omega) \\
&\ \ \ \ \ \ \ \ \ \ \ -\lim_{\mu_{p}\to-\infty}f_{p}(\omega)] \\
=&\frac{1}{h}\int_{-\infty}^{\infty}d\omega\myT_{ps}(\omega)f_{s}(\omega)\\
>& 0, 
\end{aligned}
\end{equation}
and
\begin{equation}
\begin{aligned}
\lim_{\mu_{p}\to\infty}I^{(0)}(\mu_{p},T_{p})=&\frac{1}{h}\int_{-\infty}^{\infty}d\omega\myT_{ps}(\omega)[f_{s}(\omega)\\
&\ \ \ \ \ \ \ \ \ \ \ -\lim_{\mu_{p}\to\infty}f_{p}(\omega)]\\
=&\frac{1}{h}\int_{-\infty}^{\infty}d\omega\myT_{ps}(\omega)(f_{s}(\omega)-1)\\
<&0.
\end{aligned}
\end{equation}
This ensures at least one solution due to the continuity of the currents, but does not ensure uniqueness.

We note that $I_{p}^{(0)}$ is monotonically decreasing along $\mathbf{dl}=(d\mu_{p},0)$
\begin{equation}
\Delta{I^{(0)}_{p}}=\int_{{\mu}_{p}}^{\mu_{p}'}\nabla{I^{(0)}_{p}}.\mathbf{dl}=\int_{{\mu}_{p}}^{\mu_{p}'}-\Lfun{0}{ps}d\mu_{p}
\end{equation}
due to the fact that $\Lfun{0}{ps}$ is positive, and more explicity:
\begin{equation}
\begin{aligned}
\Delta{I^{(0)}_{p}}&=\frac{1}{h}\int_{-\infty}^{\infty}d\omega\myT_{ps}(\omega)[{f}_{p}({\mu}_{p},{T}_{p};\omega)-f_{p}(\mu_{p}',{T}_{p};\omega)]\\
&>0,\ \text{if}\ \mu_{p}'<{\mu}_{p}\\
&<0,\ \text{if}\ \mu_{p}'>{\mu}_{p}.
\label{SecondLawNumberCurrent}
\end{aligned}
\end{equation}
This implies the existence of a unique solution to $I_{p}^{(0)}(\mu_{p},T_{p})=0$ for every $T_{p}\in (0,\infty)$
which we denote by $\mu_{p}=M(T_{p})$, and Eq.\ (\ref{lemma1}) is implied by Eq.\ (\ref{SecondLawNumberCurrent}).
\end{proof}
\end{LemmaSecondLawNumberCurrent}

We also note that the number current [$\mu_{p}=M(T_{p})$] contour is vertical when the temperature
approaches absolute zero, as shown in Fig.\ \ref{fig:Lemma1},
since $\Lfun{1}{ps}/T_{p}\to0$ as $T_{p}\to0$, and implies a vanishing Seebeck coefficient for the
probe-system junction near absolute zero.

An ``ideal potentiometer" was initially proposed \cite{Engquist81} by merely requiring $I^{(0)}_{p}=0$. Subsequently,
B{\"u}ttiker \cite{Buttiker1988,Buttiker89} clarified that this definition holds only near absolute zero due to the absence of thermoelectric corrections.
Such a voltage probe 
determines the voltage uniquely at zero temperature in the linear response regime,
and is relevant for experiments in mesoscopic circuits \cite{Benoit1986,Shepard1992,Picciotto2001,Gao2005} which
are carried out at cryogenic temperatures.  However, at higher temperatures and/or larger bias voltages, where the sample may be heated by both the Joule and
Peltier effects, thermoelectric corrections to voltage measurements must be considered.  Indeed, 
Bergfield and Stafford \cite{Bergfield2014} argue that an {\em ideal voltage probe}
must be required to equilibrate thermally with the system ($I^{(1)}_{p}=0$), without which ``a voltage will develop across the system-probe
junction due to the {\em Seebeck effect}." 

Voltage probes have been used extensively in the theoretical literature to mimic the effects of various scattering processes,
such as
inelastic scattering \cite{Buttiker1985b,Buttiker1986b,Buttiker1988,D'Amato1990,Ando1996,Roy2007} 
and dephasing \cite{deJong1996,vanLangen1997,Forster2007} in mesoscopic systems.
A modern variation of B{\"u}ttiker's voltage probe, additionally requiring that the probe exchange no heat current, 
has been used to model inelastic scattering in
quantum transport problems at finite temperature
\cite{Jacquet2009,Sanchez2011,Saito2011,Balachandran2013,Brandner2013}. The probe technique, as a model for scattering,
has also been extensively studied beyond the linear response regime
\cite{Bandyopadhyay2011,Bedkihal2013,Bedkihal2013b}.

Lemma \ref{LemmaSecondLawNumberCurrent1} implies that
a ``voltage probe" (defined only by $I^{(0)}_{p}=0$) requires the simultaneous
specification of a probe temperature $T_{p}$ so that $\mu_{p}=M(T_{p})$ is uniquely determined. 
Fig.\ \ref{fig:Lemma1} illustrates that the measured voltage shows a large dependence on the probe temperature.
Therefore, it is important to define a simultaneous temperature measurement by imposing $I^{(1)}_{p}(\mu_{p},T_{p})=0$.

\begin{LemmaSecondLawHeatCurrent}
\label{LemmaSecondLawHeatCurrent1}
The heat current contour defined by $I_{p}^{(1)}(\mu_{p},T_{p})=c$, where c is some constant,
obeys the second law of thermodynamics, namely,
\begin{equation}
\begin{aligned}
I^{(1)}_{p}(\mu_{p},T_{p}')&>c,\ \text{if}\ T_{p}'<{T}_{p}\\
&<c,\ \text{if}\ T_{p}'>{T}_{p}.
\label{SecondLawHeatCurrent}
\end{aligned}
\end{equation}
\end{LemmaSecondLawHeatCurrent}
\begin{proof}
We follow an analogous argument to lemma \ref{LemmaSecondLawNumberCurrent1}, and show
the monotonicity of $I_{p}^{(1)}(\mu_{p},T_{p})$ along a certain contour in the $\mu_{p}$-$T_{p}$ plane.
Naturally, the contour we choose is along a fixed $\mu_{p}$ [cf.\ Eq.\ (\ref{HeatCurrentGradient})]
since we know that $\Lfun{2}{ps}$ is positive. Therefore we have $\Delta{I^{(1)}_{p}}=I^{(1)}_{p}({\mu}_{p},{T}_{p}')-I^{(1)}_{p}({\mu}_{p},{T}_{p})
=\int_{{T}_{p}}^{T_{p}'}\nabla{I^{(1)}_{p}}.\mathbf{dl}$, where $\mathbf{dl}=(0,dT_{p})$ and explicitly,
\begin{equation}
\begin{aligned}
\Delta{I^{(1)}_{p}}&=\frac{1}{h}\int_{-\infty}^{\infty}d\omega(\omega-{\mu}_p)\myT_{ps}(\omega)[{f}_{p}({\mu}_{p},{T}_{p};\omega)\\
&\ \ \ \ \ \ \ \ \ \ \ \ \ \ \ \ \ \ \ \ \ \ \ \ \ \ \ \ \ \ \ \ \ \ \ \ \ \ \ \ -f_{p}({\mu}_{p},{T}_{p}';\omega)]\\
&>0,\ \text{if}\ T_{p}'<{T}_{p}\\
&<0,\ \text{if}\ T_{p}'>{T}_{p}.
\label{SecondLawHeatCurrent2}
\end{aligned}
\end{equation}
This implies Eq.\ (\ref{SecondLawHeatCurrent}). 
\end{proof}
\begin{figure}
\captionsetup{justification=raggedright,
singlelinecheck=false
}
\includegraphics[width=3.4in]{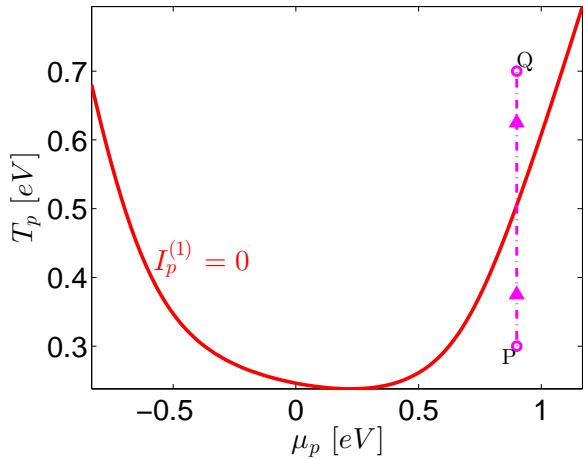}
\caption{\label{fig:Lemma2}Illustration of Lemma \ref{LemmaSecondLawHeatCurrent1}: The contour $PQ$ shown in magenta
cuts $I^{(1)}_{p}=0$ (or any $I^{(1)}_{p}=\text{constant}$) exactly once. Contour $PQ$ is defined along constant voltage $\mu_{p}=\text{constant}$, and
illustrates the Clausius statement:
The heat current is monotonically decresing along $PQ$. The system and bias conditions are detailed in Sec.\ \ref{sec:TwoLevelSystem}.
}
\end{figure}

We stated lemma \ref{LemmaSecondLawHeatCurrent1} with a constant $c$ \footnote{Furthermore, the tangent vector [cf.\ Eq.\ (\ref{HeatCurrentTV})] along $I^{(1)}_{p}=c$
cannot be of magnituge zero since $\Lfun{2}{ps}$ is strictly positive for $T_{p}\in(0,\infty)$.
Therefore, the contour $I^{(1)}_{p}=c$ doesn't terminate for finite values of $T_{p}$ and $\mu_{p}$.
This implies the existence of a function $\tau_{c}: (-\infty,\infty)\rightarrow (0,\infty)$ which 
defines
\begin{equation}
T_{p}=\tau_{c}(\mu_{p})
\end{equation}
for each point on $I^{(1)}_{p}(\mu_{p},T_{p})=c$. },
not necessarily $c=0$, unlike lemma \ref{LemmaSecondLawNumberCurrent1}. This is because we do {\em not} {\it a priori} know whether the contour
$I^{(1)}_{p}=0$ exists, and we derive a necessary and sufficient condition for its existence in Sec \ref{sec:existence}.

Analogous to Lemma \ref{LemmaSecondLawNumberCurrent1}, Lemma \ref{LemmaSecondLawHeatCurrent1} implies that a ``temperature probe" \cite{Engquist81}
(defined only by $I^{(1)}_{p}=0$) requires the simultaneous specification of a 
probe voltage $\mu_{p}$ so that the temperature $T_{p}=\tau_{0}(\mu_{p})$ (cf.\ footnote \footnotemark[2]) is uniquely determined.
Fig.\ \ref{fig:Lemma2} illustrates that the measured temperature shows a large dependence on the probe voltage. Therefore, it becomes
important to simultaneously measure the voltage by imposing $I^{(0)}_{p}=0$.
If the temperature probe is not allowed to equilibrate electrically with the system, then a temperature difference will build up across the probe-system
junction due to the {\em Peltier effect}, leading to an error in the temperature measurement.

Clearly, depending upon the probe voltage, the ``temperature probe" could measure any of a range of values, rendering
the measurement somewhat meaningless (see Fig.\ \ref{fig:Lemma2}). 
Analogously, the ``voltage probe" could measure any of a range of values depending upon the probe temperature (see Fig.\ \ref{fig:Lemma1}).
{\em Thermoelectric probes} (also referred to as dual probes, and voltage-temperature probes) treat
temperature and voltage measurements on an equal footing, and implicitly account for the thermoelectric corrections exactly.
Only such a dual probe is in {\em both thermal and electrical equilibrium} with the system being measured, and therefore yields an unbiased measurement
of both quantities.
A mathematical proof of the uniqueness of a voltage and temperature measurement is 
therefore of fundamental importance.

We may also deduce that $T_{p}=0$ cannot be obtained as a measurement outcome since
\begin{equation}
\label{ThirdLaw}
\begin{aligned}
\lim_{{T_{p}}\to{0}}I^{(1)}_{p}(\mu_{p},T_{p})&= \int_{-\infty}^{\infty}d\omega\ (\omega-\mu_{p})\myT_{ps}(\omega)[f_{s}(\omega)\\
&\ \ \ \ \ \ \ \ \ \ \ \ \ \ \ \ \ \ \ \ \ -\lim_{T_{p}\to0}f_{p}(\mu_{p},T_{p})]\\
&=\int_{-\infty}^{\mu_{p}}d\omega\ (\omega-\mu_{p})\myT_{ps}(\omega)(f_{s}(\omega)-1)\\ &\ \ \ \ + \int_{\mu_{p}}^{\infty}d\omega\ (\omega-\mu_{p})\myT_{ps}(\omega)f_{s}(\omega)\\
&>0,
\end{aligned}
\end{equation}
{\em consistent with the third law of thermodynamics}. However, temperatures arbitarily close to absolute zero are, in principle, possible \cite{Shastry2015}. 

Lemmas \ref{LemmaSecondLawNumberCurrent1} and \ref{LemmaSecondLawHeatCurrent1} are equivalent to the Clausius statement of the second law 
\cite{Clausius1854}: 
``No process is possible whose sole effect is to transfer heat from a colder body to a warmer body.''
Lemma \ref{LemmaSecondLawHeatCurrent1} gives us the direction in which heat will flow [cf.\ Eq.\ (\ref{SecondLawHeatCurrent2})] when the probe is biased away from thermal equilibrium $I^{(1)}_{p}(\mu_{p},T_{p})=0$:
whenever the probe is hotter than the temperature corresponding to local thermal equilibrium, with the chemical potential held constant, heat flows {\em out} of the probe and vice versa.
Similarly, Lemma \ref{LemmaSecondLawNumberCurrent1} gives us the direction in which particle flow occurs when the probe is biased away from electrical equilibrium $I^{(0)}_{p}(\mu_{p},T_{p})=0$:
whenever the probe is at a higher chemical potential than the one corresponding to local electrical equilibrium, with temperature held constant, particles flow {\em out} of the probe and vice versa.

The problem of a unique measurement of a ``voltage probe" (defined only by $I^{(0)}_{p}=0$), or a ``temperature probe" (defined only by $I^{(1)}_{p}=0$)
has been attempted previously by Jacquet and Pillet \cite{Jacquet2012} for transport beyond linear response, and to our knowledge
is the only work in this direction.
However, in Ref.\ \cite{Jacquet2012}, the bias conditions considered are quite restrictive and the result assumes noninteracting electrons.
Lemma \ref{LemmaSecondLawNumberCurrent1} and lemma \ref{LemmaSecondLawHeatCurrent1}, respectively, generalize the result to 
arbitrary bias conditions, and arbitary interactions within a quantum electron system while also
providing a useful insight via the Clausius statement of the second law of thermodynamics.
However, the question we would like to answer in this article pertains to the uniqueness of a {\em thermoelectric probe} measurement, defined by both $I^{(0)}_{p}=0$
and $I^{(1)}_{p}=0$. A result for such dual probes has been obtained only in the linear response regime
and for noninteracting electrons \cite{Jacquet2009}.

\begin{TheoremOnsager}
\label{ThmOnsager}
The coefficients $\Lfun{\nu}{ps}$ satisfy the inequality
\begin{equation}
\Lfun{0}{ps}\Lfun{2}{ps}-\big(\Lfun{1}{ps}\big)^{2}>0.
\end{equation}
\end{TheoremOnsager}
\begin{proof}
We may define functions $g(\omega)$ and $h(\omega)$ as
\begin{equation}
g(\omega)=\sqrt{\myT_{ps}(\omega)\bigg(-\frac{\partial{f_{p}}}{\partial{\omega}}\bigg)}
\end{equation}
and
\begin{equation}
h(\omega)=(\omega-\mu_{p})\sqrt{\myT_{ps}(\omega)\bigg(-\frac{\partial{f_{p}}}{\partial{\omega}}\bigg)}.
\end{equation}
We note that $g(\omega)$ and $h(\omega)$ belong to $\mathbf{L}^{2}(\mathbb{R})$ \footnotemark[1].
Noting that $g$ and $h$ are real, we apply the Cauchy-Schwarz inequality
\begin{equation}
\bigg|\int_{-\infty}^{\infty}d\omega g(\omega)h(\omega)\bigg|^{2}\leq \int_{-\infty}^{\infty}d\omega|g(\omega)|^{2}\int_{-\infty}^{\infty}d\omega|h(\omega)|^{2}.
\end{equation}
The integral appearing on the $lhs$ is $\Lfun{1}{ps}$, while on the $rhs$ we have the product of $\Lfun{0}{ps}$
and $\Lfun{2}{ps}$, respectively. We drop the absolute value on the $lhs$ by noting that $\Lfun{1}{ps}$ is real and write 
\begin{equation}
\big(\Lfun{1}{ps}\big)^{2}\leq\Lfun{0}{ps}\Lfun{2}{ps}.
\end{equation}
We drop the equality case above by noting that $g$ and $h$ are linearly independent except for the trivial case when
$\myT_{ps}(\omega)=0\ \forall \omega$, or when the probe coupling is narrowband
[$\myT_{ps}(\omega)=\bar{\gamma}\delta(\omega-\omega_{0})$] which we discuss in sec \ref{sec:Narrowband}.
\end{proof}

The proof above can be easily extended to show the positive-definiteness
of the linear response matrices \cite{Onsager1931}  
widely used for elastic 
transport calculations (e.g., in Refs.\ \onlinecite{Sivan1986,Bergfield10b}). Theorem \ref{ThmOnsager} implies a positive thermal conductance 
(see e.g., Ref.\ \cite{Bergfield10b}), which is necessary for positive entropy production consistent with the second law of thermodynamics.

\subsection{Uniqueness}

\label{subsec:Uniqueness}
\begin{TheoremUniqueness}
\label{ThmUniqueness}
The local temperature and voltage of a nonequilibrium quantum system, measured by a {\em thermoelectric probe}, is unique when it exists.
\end{TheoremUniqueness}
\begin{proof}
The tangent vectors $\mathbf{t}^{(\nu)}$ for $I^{(\nu)}_{p}$ are along
\begin{equation}
\mathbf{t}^{(0)}=\bigg(-\frac{\Lfun{1}{ps}}{T_{p}},\Lfun{0}{ps}\bigg)
\end{equation}
and
\begin{equation}
\label{HeatCurrentTV}
\begin{aligned}
\mathbf{t}^{(1)}&=\bigg(\frac{\Lfun{2}{ps}}{T_{p}},-\Lfun{1}{ps}-I^{(0)}_{p}\bigg)\\
&=\bigg(\frac{\Lfun{2}{ps}}{T_{p}},-\Lfun{1}{ps}\bigg),\ \text{if}\ I^{(0)}_{p}=0,
\end{aligned}
\end{equation}
respectively, such that we have
\begin{equation}
\int_{s_{1}}^{s_{2}}d{s}\frac{\mathbf{t}^{(\nu)}\cdot\nabla{I^{(\nu)}_{p}}}{|\mathbf{t}^{(\nu)}|}=0,
\end{equation}
where $s$ is a scalar that labels points along the contour $I^{(\nu)}_{p}=\text{constant}$.

We now compute the change in $I^{(1)}_{p}$ along the contour $I^{(0)}_{p}=0$. The points along $I^{(0)}_{p}=0$ are labeled
by the continuous parameter $\xi$ such that $\mu_{p}=\mu_{p}(\xi)$ and $T_{p}=T_{p}(\xi)$. $\xi$ is chosen to be 
increasing with increasing temperature. The change $\Delta{I^{(1)}_{p}}$ becomes
\begin{equation}
\begin{aligned}
\label{eqn:Onsager}
\Delta{I^{(1)}_{p}}&=\int_{\xi_{1}}^{\xi_{2}}d\xi\frac{\mathbf{t}^{(0)}\cdot\nabla{I^{(1)}_{p}}}{|\mathbf{t}^{(0)}|}\\
&=\int_{\xi_{1}}^{\xi_{2}}d\xi\frac{1}{|\mathbf{t}^{(0)}|T_{p}}\big((\Lfun{1}{ps})^{2}-\Lfun{0}{ps}\Lfun{2}{ps}\big)\\
&>0\ \text{if}\ \xi_{2}<\xi_{1}\\
&<0\  \text{if}\ \xi_{2}>\xi_{1},
\end{aligned}
\end{equation}
due to theorem \ref{ThmOnsager}. Therefore $I^{(1)}_{p}=0$ (or for that matter $I^{(1)}_{p}=c$, for any $c$) 
is satisfied at most at a single point along $I^{(0)}_{p}=0$.
\end{proof}

Theorem \ref{ThmOnsager} is a form of the second law of thermodynamics that gives us the direction in which the heat current flows
along the contour $I^{(0)}_{p}=0$ (cf.\ Eq.\ (\ref{eqn:Onsager})). The heat current
$I^{(1)}_{p}$ decreases monotonically along the contour $I^{(0)}_{p}=0$.
Therefore we may find only one point along $I^{(0)}_{p}=0$
that also satisfies $I^{(1)}_{p}=0$ , which implies a unique solution to Eq.\ (\ref{equilibrium}) when it exists.

Indeed, Onsager points out in his 1931 paper \cite{Onsager1931} that for positive entropy production,
the linear response matrix will have to be positive-definite (which translates to our condition in Theorem \ref{ThmOnsager}).
However, that analysis rests upon the assumption of linear response near equilibrium. 
Our result in Theorem \ref{ThmOnsager} does not require such
a condition for the nonequilibrium state of the system, but instead emerges out of the analysis of the currents flowing into a
weakly-coupled probe. In addition, we obtain a strict mathematical proof of theorem \ref{ThmOnsager}. We point out that theorem \ref{ThmOnsager}
holds even when the physically expected postulate \ref{postulate} fails, making the uniqueness result in theorem \ref{ThmUniqueness}
very general \footnotemark[1].

\begin{figure*}[tbh]
\centering
\captionsetup{justification=raggedright,
singlelinecheck=false
}
{\includegraphics[width=3.4in]{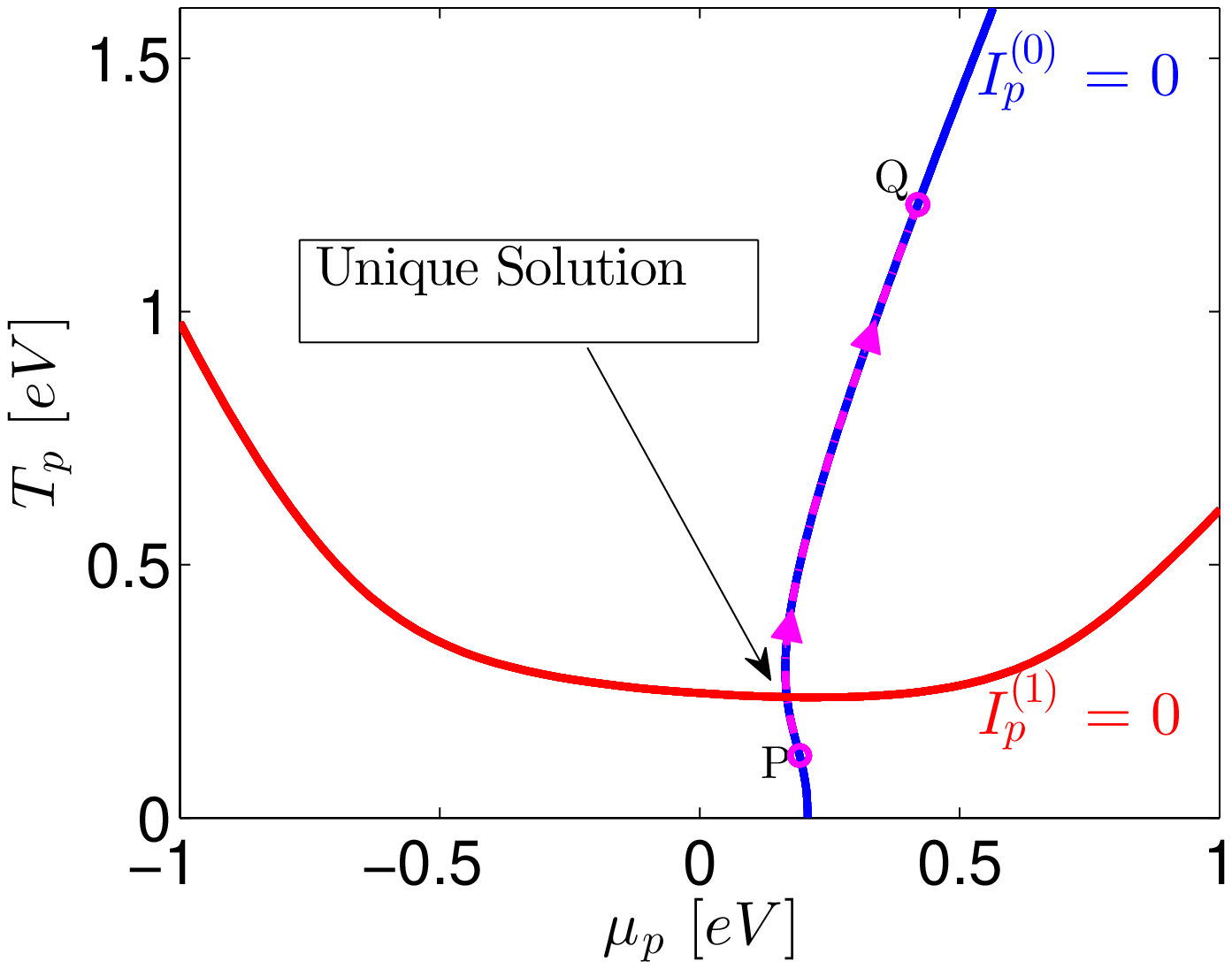}}
{\includegraphics[width=3.4in]{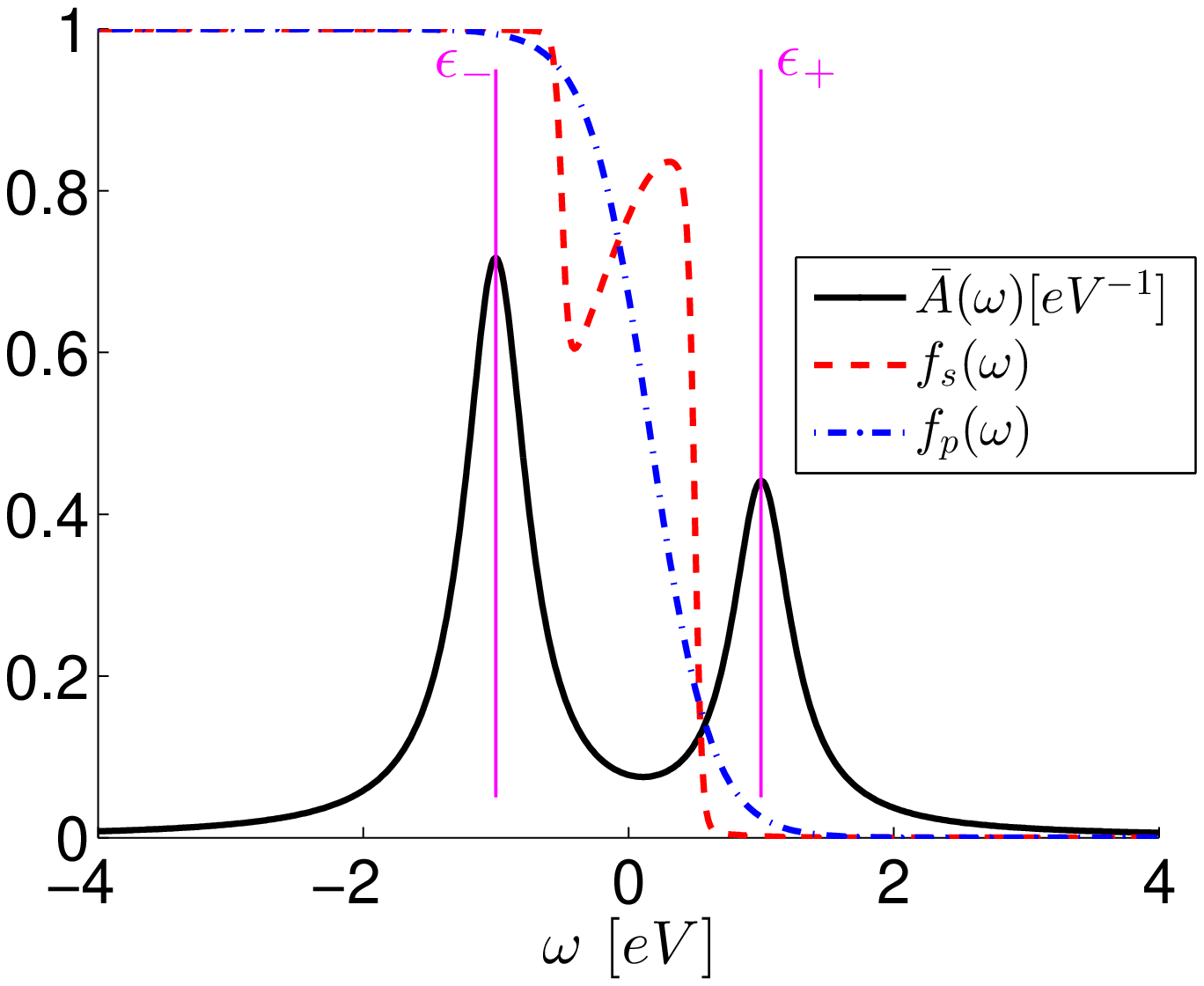}}
\caption{Left panel: Illustration of Theorem \ref{ThmUniqueness} for positive temperatures. The contour $PQ$
along $I^{(0)}_{p}=0$ (shown in blue) cuts the contour $I^{(1)}_{p}=0$ (shown in red) exactly once.
Contour $PQ$ illustrates a certain statement of the second law of thermodynamics: The heat current is monotonically decreasing along $PQ$
(thus implying uniqueness).
Right panel: The local spectrum sampled by the probe $\bar{A}(\omega)$ (black), the nonequilibrium distribution function $f_{s}(\omega)$ (red),
and the probe Fermi-Dirac distribution $f_{p}(\omega)$ (blue) corresponding to the unique solution in the left panel. 
The resonances in the spectrum $\bar{A}(\omega)$ correspond to the eigenstates
of the closed two-level Hamiltonian (see Sec. \ref{sec:TwoLevelSystem}) $\epsilon_{\pm}=\pm 1$ shown in magenta. 
The Fermi-Dirac distribution
is monotonically decreasing with energy, and corresponds to a situation with positive temperature (no net population inversion).
The necessary and sufficient condition for the existence of a positive temperature solution is stated in Theorem \ref{exists}.
}
\label{fig:Theorem2}

\end{figure*}

\section{Existence}
\label{sec:existence}

A unique local measurement of temperature and voltage is only part of our main problem. An equally important part is to derive the conditions
for the existence of a solution. The main idea behind this analysis is to follow the number current contour $I^{(0)}_{p}=0$ and
ask what happens to the heat current $I^{(1)}_{p}$ as we traverse towards higher and higher temperatures, $T_{p}\to\infty$. 
We noted that near $T_p=0$,
the heat current into the probe must be positive, consistent with the third law of thermodynamics [cf.\ Eq.\ (\ref{ThirdLaw})].
Since we know that the heat current is monotonically decreasing along the number current contour (Theorem \ref{ThmUniqueness}), 
we could guess whether or not a solution occurs
depending upon the asymptotic value of the heat current along that contour as $T_{p}\to\infty$. In this way, we find a necessary and sufficient condition
for the existence of a solution while analyzing the problem for positive temperatures (see Fig.\ \ref{fig:Theorem2} for an illustration of this case).
On the other hand, when this condition is not met, one can 
immediately prove that a negative temperature must satisfy the measurement condition $I^{(\nu)}_{p}=0,\ \nu=\{0,1\}$. This latter condition corresponds
to a system exhibiting local population inversion which leads to negative temperature \cite{Pathria2011}
solutions, as illustrated in Fig.\ \ref{fig:Corollary3}.

Our results here are again completely general and are valid for electron systems with arbitrary interactions, arbitrary steady state bias conditions,
and for any weakly-coupled probe. However, our analysis here leads us to demarcate
between two extremes of the probe-system coupling. We conclude that an {\em ideal probe} is one which operates in the {\em broadband limit}. A measurement by
such a probe depends only on the properties of the system that it couples to, and is independent of the spectral properties of the probe itself.
The broadband limit lends itself to an easier physical interpretation of the population inversion condition as well, and we discuss this important limit
in Sec.\ \ref{subsec:Broadband}.
The other extreme is that of a {\em narrowband probe} which is capable of probing the system at just one value of energy, 
leading to a 
nonunique measurement (see also the proof of theorem \ref{ThmOnsager}),
and is discussed in Sec.\ \ref{sec:Narrowband}. Only this pathological case leads to an exception to Theorem \ref{ThmUniqueness}.

The simplest system which could, in principle, exhibit population inversion is a two-level system. Therefore, our results,
including that of the previous section, have been illustrated by using a two-level system. The details of the nonequilibrium
two-level system and its coupling to the thermoelectric probe are given in Sec.\ \ref{sec:TwoLevelSystem}.

Our analysis starts with a rearragement of the currents given by 
Eq.\ (\ref{Rearraged}) and a restatement of the measurement condition [cf.\ Eq.\ (\ref{equilibrium})]
in terms of energy currents, and we
also define some useful quantities along the way. We may rewrite the number current in Eq.\ (\ref{Rearraged})
as
\begin{equation}
I_{p}^{(0)}=\langle\dot{N}\rangle{|}_{f_{s}}-\langle\dot{N}\rangle{|}_{f_{p}}
\end{equation}
where
\begin{equation}
\langle\dot{N}\rangle{|}_{f_{s}}\equiv\frac{1}{h}\int_{-\infty}^{\infty}d\omega\myT_{ps}(\omega)f_{s}(\omega),
\label{N_fs}
\end{equation}
and similarly
\begin{equation}
\langle\dot{N}\rangle{|}_{f_{p}}\equiv\frac{1}{h}\int_{-\infty}^{\infty}d\omega\myT_{ps}(\omega)f_{p}(\omega).
\label{N_fp}
\end{equation}
The quantitity $\langle\dot{N}\rangle{|}_{f_{s}}$ is the
rate of particle flow into the probe from the system, while $\langle\dot{N}\rangle{|}_{f_{p}}$ gives the rate of particle flow out of the probe
and into the system.

Similarly, the rate of energy flow into the probe from the system is
\begin{equation}
\langle\dot{E}\rangle{|}_{f_{s}}\equiv\frac{1}{h}\int_{-\infty}^{\infty}d\omega\ \omega\myT_{ps}(\omega)f_{s}(\omega),
\label{E_fs}
\end{equation}
while
\begin{equation}
\langle\dot{E}\rangle{|}_{f_{p}}\equiv\frac{1}{h}\int_{-\infty}^{\infty}d\omega\ \omega\myT_{ps}(\omega)f_{p}(\omega)
\label{E_fp}
\end{equation}
gives the rate of energy outflux from the probe back into the system. The net energy current flowing into the probe
is given by $I^{E}_{p}=\langle\dot{E}\rangle{|}_{f_{s}}-\langle\dot{E}\rangle{|}_{f_{p}}$.

The local equilibration conditions in Eq.\ (\ref{equilibrium}) now become
\begin{equation}
\begin{aligned}
\langle\dot{N}\rangle{|}_{f_{p}}&=\langle\dot{N}\rangle{|}_{f_{s}}\\
\langle\dot{E}\rangle{|}_{f_{p}}&=\langle\dot{E}\rangle{|}_{f_{s}}.
\end{aligned}
\label{NewEquilibrium}
\end{equation}
The equation for the rate of energy flow above is equivalent to the condition $I_{p}^{(1)}=0$ when $I_{p}^{(0)}=0$ since
\begin{equation}
I^{E}_{p}(\mu_{p},T_{p})\equiv\langle\dot{E}\rangle{|}_{f_{s}}-\langle\dot{E}\rangle{|}_{f_{p}}= I_{p}^{(1)} + \mu_{p}I_{p}^{(0)}.
\label{EnergyCurrent}
\end{equation}
The $lhs$ in Eq.\ (\ref{NewEquilibrium})
depends upon the probe parameters (temperature and voltage) while the $rhs$ is fixed for a given nonequilibrium system
with a given local distribution function $f_{s}(\omega)$.
The probe measures the appropriate
voltage and temperature when it exchanges no net charge and energy with the system. 

We may introduce a characteristic rate of particle flow [cf.\ Eq.\ (\ref{finiteParticleRate})] as
\begin{equation}
\begin{aligned}
\langle\dot{N}\rangle{|}_{f\equiv1}&=\frac{1}{h}\int_{-\infty}^{\infty}d\omega\myT_{ps}(\omega)\\
&\equiv\frac{\gamma_{p}}{\hbar}.
\end{aligned}
\label{ParticleRate}
\end{equation}
This leads to the following inequalities:
\begin{equation}
\begin{aligned}
0<&\langle\dot{N}\rangle{|}_{f_{s}}<\frac{\gamma_{p}}{\hbar},\\
0<&\langle\dot{N}\rangle{|}_{f_{p}}<\frac{\gamma_{p}}{\hbar}.
\end{aligned}
\label{ParticleInequality}
\end{equation}
The $lhs$ in the inequality for $\langle\dot{N}\rangle{|}_{f_{s}}$ above excludes $f_{s}(\omega)\equiv0$ while the $rhs$ excludes $f_{s}(\omega)=1\ \forall \omega\in\mathbb{R}$,
and we retain the strict inequalities imposed by Eq.\ (\ref{ParticleInequality}) (see also Eqs.\ (\ref{BoundedParticleRates}) and (\ref{BoundedEnergyRates})
and the preceding discussion).

We similarly introduce a characteristic rate for the energy flow between the system and probe:
\begin{equation}
\begin{aligned}
\langle\dot{E}\rangle{|}_{f\equiv1}&=\frac{1}{h}\int_{-\infty}^{\infty}d\omega\ \omega\myT_{ps}(\omega)\\
&\equiv\frac{\gamma_{p}}{\hbar}\omega_{c},
\end{aligned}
\label{EnergyRate}
\end{equation}
where ${\omega_{c}}<\infty$ (due to postulate \ref{postulate}) 
can interpreted as the centroid of the probe-sample transmission function. We find that $\omega_{c}\to\infty$ necessarily implies a positive temperature
solution. We remind the reader that $\omega_{c}\to-\infty$ is physically impossible due to the principle that
any physical system must have a lower bound for the energy ($\langle H\rangle\geq-c$ for some finite $c\in\mathbb{R}$).

The quantities $\langle\dot{N}\rangle{|}_{f_{s}}$, $\langle\dot{N}\rangle{|}_{f_{p}}$, $\langle\dot{N}\rangle{|}_{f\equiv1}$,
$\langle\dot{E}\rangle{|}_{f_{s}}$,  $\langle\dot{E}\rangle{|}_{f_{p}}$, $\langle\dot{E}\rangle{|}_{f\equiv1}$
are all finite due to postulate \ref{postulate} [cf.\ Eqs.\ (\ref{finiteParticleRate}-\ref{BoundedEnergyRates})].

\subsection{Asymptotic Properties, and Conditions for the Existence of a Solution}

Traversing along $I^{(0)}_{p}=0$ results in a monotonically {\em decreasing} heat current $I^{(1)}_{p}$ (Theorem \ref{ThmUniqueness}).
Here, we traverse the contour from low temperatures ($T_{p}\to0$) to higher temperatures ($T_{p}\to\infty$)
as discussed in Theorem \ref{ThmUniqueness}. This implies
a monotonically {\em increasing} ${\langle\dot{E}\rangle{|}_{f_{p}}}$ due to Eq.\ (\ref{EnergyCurrent}).
We proceed to calculate the asymtotic value of ${\langle\dot{E}\rangle{|}_{f_{p}}}$ along the number current contour.

Let the asymptotic scaling of $\mu_{p}=M(T_{p})$ defined by the contour $I^{(0)}_{p}(\mu_{p},T_{p})=0$ (lemma \ref{LemmaSecondLawNumberCurrent1}) be
\begin{equation}
\lim_{T_{p}\to\infty}\frac{M(T_{p})}{T_{p}}=\Lambda.
\label{AsympoticValue}
\end{equation}
We use the above limiting value to calculate $\langle\dot{N}\rangle{|}_{f_{p}}$ along the contour $\mu_{p}=M(T_{p})$:
\begin{equation}
\begin{aligned}
\lim_{T_{p}\to\infty}\langle\dot{N}\rangle{|}_{f_{p}}&=\frac{1}{h}\int_{-\infty}^{\infty}d\omega\myT_{ps}(\omega)\\
&\ \ \ \ \ \ \ \ \ \ \ \ \ \ \ \ \times\lim_{T_{p}\to\infty}\frac{1}{1+\exp{\bigg({\frac{\omega-M(T_{p})}{T_{p}}}\bigg)}}\\
&=\frac{1}{h}\int_{-\infty}^{\infty}d\omega\myT_{ps}(\omega)\frac{1}{1+\exp{(-\Lambda)}}\\
&=\frac{1}{1+\exp{(-\Lambda)}}\frac{\gamma_{p}}{\hbar}.
\end{aligned}
\end{equation}
The above limiting value satisfies the inequality in Eq.\ (\ref{ParticleInequality}) for any $\Lambda\in\mathbb{R}$.
The points on the contour satisfy $\langle\dot{N}\rangle{|}_{f_{p}}=\langle\dot{N}\rangle{|}_{f_{s}}$ by construction, therefore
$\Lambda$ is computed from the equation 
\begin{equation}
\frac{1}{1+\exp{(-\Lambda)}}\frac{\gamma_{p}}{\hbar}=\langle\dot{N}\rangle{|}_{f_{s}}.
\label{LambdaValue}
\end{equation}

It is important to note that the asymptotic scaling defined by Eq.\ (\ref{AsympoticValue}) does not mean that the scaling is linear.
For example, a sublinear scaling $M(T_{p})=\alpha T_{p}^{n}$ with $n<1$ merely corresponds to $\Lambda=0$ which
could satisfy Eq.\ (\ref{LambdaValue}) if the nonequilibrium system is prepared in that way. 
However, $\Lambda\to\pm\infty$ do not obey the strict inequality in Eq.\ (\ref{ParticleInequality}). $\Lambda\to\infty$ corresponds to
a trivial and unphysical nonequilibrium distribution
$f_{s}(\omega)\equiv1$, and likewise, $\Lambda\to-\infty$ corresponds to $f_{s}(\omega)\equiv0\ \forall\omega$.

The asymtotic value of $\langle\dot{E}\rangle{|}_{f_{p}}$ along the $I^{(0)}_{p}=0$ contour is simply
\begin{equation}
\label{AsymptoticEnergy}
\begin{aligned}
\lim_{T_{p}\to\infty}\langle\dot{E}\rangle{|}_{f_{p}}&=\frac{1}{h}\int_{-\infty}^{\infty}d\omega\ \omega\myT_{ps}(\omega)\\
&\ \ \ \ \ \ \ \ \ \ \ \ \ \ \ \times\lim_{T_{p}\to\infty}\frac{1}{1+\exp\bigg({\frac{\omega-M(T_{p})}{T_{p}}}\bigg)}\\
&=\frac{1}{h}\int_{-\infty}^{\infty}d\omega\ \omega\myT_{ps}(\omega)\frac{1}{1+\exp{(-\Lambda)}}\\
&=\frac{1}{1+\exp{(-\Lambda)}}\frac{\gamma_{p}}{\hbar}\omega_{c}\\
&=\omega_{c}\langle\dot{N}\rangle{|}_{f_{s}}.
\end{aligned}
\end{equation}

\begin{TheoremExistence}
\label{exists}
A positive temperature solution exists if and only if there is no net population inversion, i.e., when
\begin{equation}
\frac{\langle\dot{E}\rangle{|}_{f_{s}}}{\langle\dot{N}\rangle{|}_{f_{s}}}<\omega_{c}.
\end{equation}
\end{TheoremExistence}
\begin{proof}
$\langle\dot{E}\rangle{|}_{f_{p}}/\langle\dot{N}\rangle{|}_{f_{s}} < \langle\dot{E}\rangle{|}_{f_{s}}/\langle\dot{N}\rangle{|}_{f_{s}}$
when $T_{p}\rightarrow0$ along the contour $I^{(0)}_{p}=0$ [cf.\ Eq.\ (\ref{ThirdLaw})
and Eq.\ (\ref{EnergyCurrent})]. The asymptotic limit of $\langle\dot{E}\rangle{|}_{f_{p}}/\langle\dot{N}\rangle{|}_{f_{s}}$
is $\omega_{c}$ [cf. Eq.\ (\ref{AsymptoticEnergy})]. $\langle\dot{E}\rangle{|}_{f_{p}}$ is 
continuous $\forall\ \ \mu_{p}\in(-\infty,\infty), T_{p}\in(0,\infty)$ and is monotonically increasing along $I^{(0)}_{p}=0$ (Theorem \ref{ThmUniqueness}).
We use the intermediate value theorem.
\end{proof}

\begin{figure*}
\captionsetup{justification=raggedright,
singlelinecheck=false
}
\includegraphics[width=3.56in]{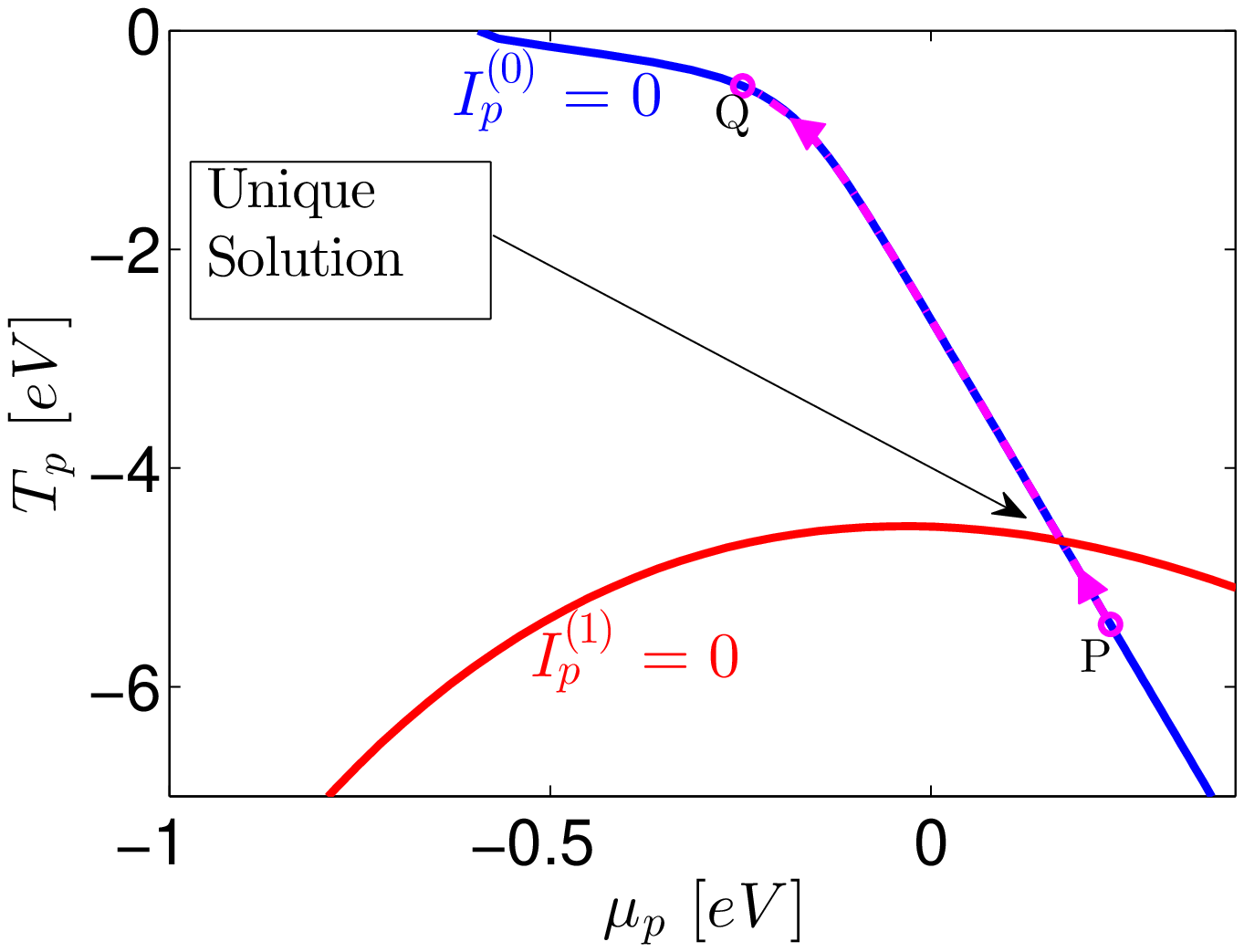}
\includegraphics[width=3.4in]{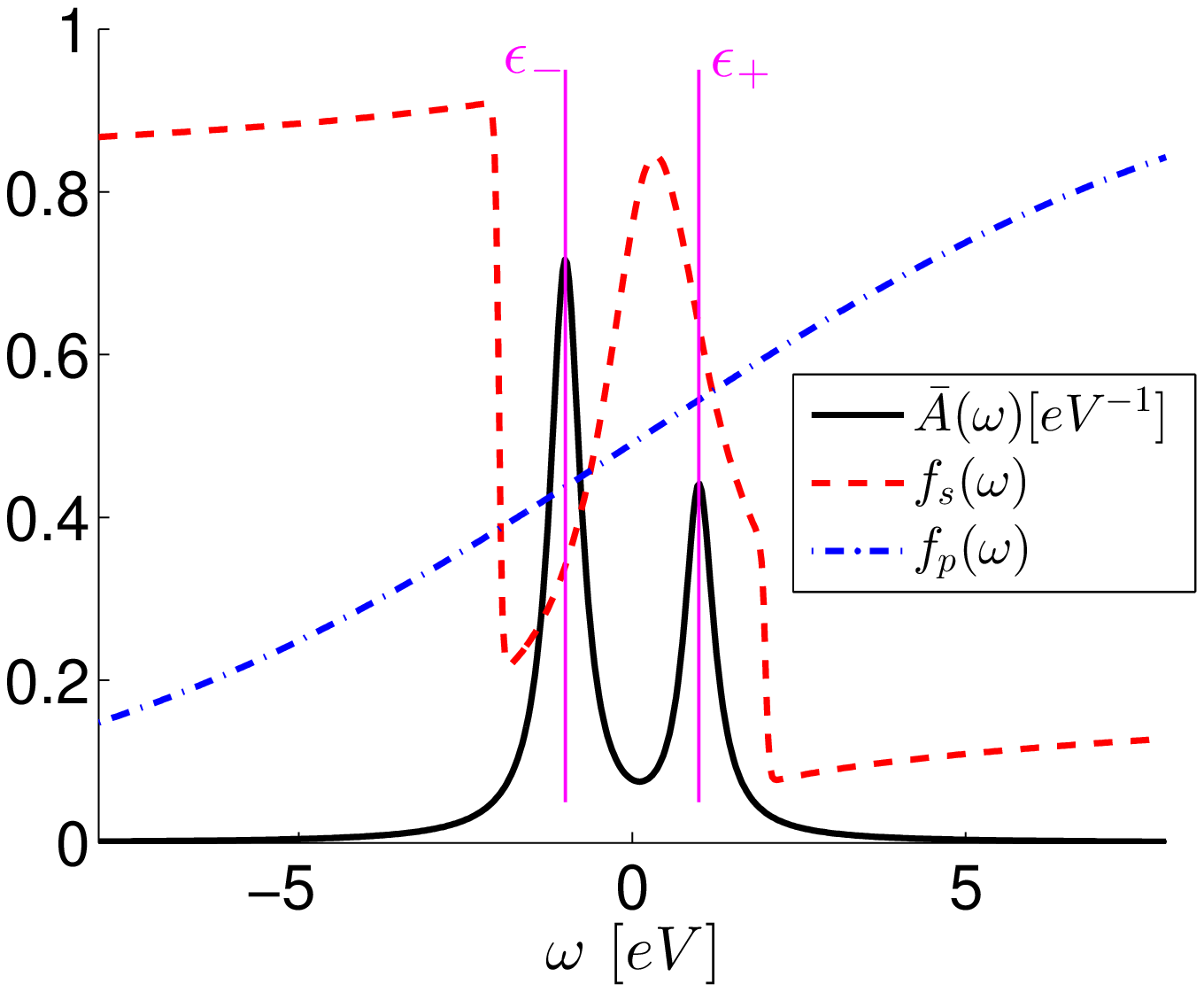}
\caption{\label{fig:Corollary3}Left panel: Illustration of Theorem \ref{ThmUniqueness} for negative temperatures. The contour $PQ$
along $I^{(0)}_{p}=0$ (shown in blue) cuts the contour $I^{(1)}_{p}=0$ (shown in red) exactly once.
Contour $PQ$ illustrates a certain statement of the second law of thermodynamics: The heat current is monotonically decreasing along $PQ$
(thus implying uniqueness).
Right panel: The local spectrum sampled by the probe $\bar{A}(\omega)$ (black, and nearly unchanged from Fig. \ref{fig:Theorem2}), 
the nonequilibrium distribution function $f_{s}(\omega)$ (red),
and the probe Fermi-Dirac distribution $f_{p}(\omega)$ (blue) which corresponds to the unique solution (shown in the left panel). 
The resonances in the spectrum $\bar{A}(\omega)$ correspond to the eigenstates
of the closed two-level Hamiltonian (see Sec. \ref{sec:TwoLevelSystem}) $\epsilon_{\pm}=\pm 1$ shown in magenta. 
The system has a net population inversion, satisfying the conditions of Corollary \ref{existence:Corollary}, and the probe
Fermi-Dirac distribution
is monotonically increasing with energy, corresponding to a negative temperature. 
}
\end{figure*}

\begin{CorollaryExistence}
\label{existence:Corollary}
There exists a negative temperature solution for a nonequilibrium system with net population inversion, i.e., when
\begin{equation}
\frac{\langle\dot{E}\rangle{|}_{f_{s}}}{\langle\dot{N}\rangle{|}_{f_{s}}}>\omega_{c}.
\end{equation}
\end{CorollaryExistence}
\begin{proof}
Let $f_{p}(\mu_{p},T_{p})$ be the Fermi-Dirac distribution with $T_{p}>0$; we
define the Fermi-Dirac distribution $f_{p}^{-}\equiv f_{p}(\mu_{p},-T_{p})=1-f_{p}$.

\begin{equation}
\begin{aligned}
I_{p}^{(\nu)}(\mu_{p},-T_{p})&=\frac{1}{h}\int_{-\infty}^{\infty}d\omega(\omega-\mu_{p})^{\nu}\myT_{ps}(\omega)[f_{s}(\omega)\\
&\ \ \ \ \ \ \ \ \ \ \ \ \ \ \ \ \ \ \ \ \ \ \ \ \ \ \ \ \ \ \ \ \  -\big(1-f_{p}(\omega)\big)]\\
&=\frac{1}{h}\int_{-\infty}^{\infty}d\omega(\omega-\mu_{p})^{\nu}\myT_{ps}(\omega)[f_{p}(\omega)\\
&\ \ \ \ \ \ \ \ \ \ \ \ \ \ \ \ \ \ \ \ \ \ \ \ \ \ \ \ \ \ \ \ \  -\big(1-f_{s}(\omega)\big)]\\
&=\frac{1}{h}\int_{-\infty}^{\infty}d\omega(\omega-\mu_{p})^{\nu}\myT_{ps}(\omega)[f_{p}(\omega)-f_{s}^{-}(\omega)]\\
&\equiv -I_{p}^{(\nu)-}
\label{NegativeCurrents}
\end{aligned}
\end{equation}
$I_{p}^{(\nu)-}=0$ with $\nu=\{0,1\}$ is now understood to solve the complementary nonequilibrium system with $f_{s}^{-}(\omega)\equiv1-f_{s}(\omega)$.

$f_{s}^{-}(\omega)$ is of course a completely valid nonequilibrium distribution function and satisfies Eq.\ (\ref{fsBounded}).
We apply Theorem \ref{exists} and find that
\begin{equation}
\begin{aligned}
\langle\dot{E}\rangle{|}_{f_{s}^{-}}&<\ \omega_{c}\langle\dot{N}\rangle{|}_{f_{s}^{-}}\\
\frac{\gamma_{p}}{\hbar}\omega_{c}-\langle\dot{E}\rangle{|}_{f_{s}}&<\ \omega_{c}\big(\frac{\gamma_{p}}{\hbar}-\langle\dot{N}\rangle{|}_{f_{s}}\big)\\
-\langle\dot{E}\rangle{|}_{f_{s}}&<\ -\omega_{c}\langle\dot{N}\rangle{|}_{f_{s}}\\
\langle\dot{E}\rangle{|}_{f_{s}}&>\ \omega_{c}\langle\dot{N}\rangle{|}_{f_{s}}.
\end{aligned}
\end{equation}
For the case that $\langle\dot{E}\rangle{|}_{f_{s}}=\ \omega_{c}\langle\dot{N}\rangle{|}_{f_{s}}$, $T_p=\pm \infty$, corresponding to $f_p =1/2$,
independent of energy.
\end{proof}

\subsection{Ideal Probes: The Broadband Limit}
\label{subsec:Broadband}
In the broadband limit, the probe-system coupling becomes energy independent, and we may write $\Gamma^{p}(\omega)=\Gamma^{p}(\mu_{0})$. 
The spectrum of the system,
sampled locally by the probe, is given by
\begin{equation}
\label{LocalSpectrum}
\begin{aligned}
\bar{A}(\omega)\equiv& \frac{\Tr{\Gamma^{p}(\omega)A(\omega)}}{\Tr{\Gamma^{p}(\omega)}}\\
=&\frac{\Tr{\Gamma^{p}(\mu_{0})A(\omega)}}{\Tr{\Gamma^{p}(\mu_{0})}}.
\end{aligned}
\end{equation}

The occupancy and energy of the system, respectively, are given by
\begin{equation}
\begin{aligned}
\langle{N}\rangle{|}_{f_{s}}=& \int_{-\infty}^{\infty}d\omega\bar{A}(\omega)f_{s}(\omega)\\
\langle{E}\rangle{|}_{f_{s}}=& \int_{-\infty}^{\infty}d\omega\ \omega\bar{A}(\omega)f_{s}(\omega).
\end{aligned}
\label{eq:moments}
\end{equation}
The measurement conditions in Eq.\ (\ref{equilibrium}) become simply \cite{Stafford2014}
\begin{equation}
\label{BBMeasurement}
\begin{aligned}
\langle{N}\rangle{|}_{f_{p}}&=\langle{N}\rangle{|}_{f_{s}}\\
\langle{E}\rangle{|}_{f_{p}}&=\langle{E}\rangle{|}_{f_{s}}.
\end{aligned}
\end{equation}
The above equations imply that an {\em ideal measurement} of voltage and temperature 
constitutes a measurement of the zeroth and first moments of the local energy distribution of the system. That is to say, 
when the probe is in local equilibrium with the nonequilibrium system,
the local occupancy and energy of the system are the same as they would be if the system's local spectrum were populated by
the equilibrium Fermi-Dirac distribution $f_{p}\equiv f_{p}(\mu_{p},T_{p})$ of the probe. 


We may now write the condition for the existence of a positive temperature solution (Theorem \ref{exists}) simply as
\begin{equation}
\frac{\langle{E}\rangle{|}_{f_{s}}}{\langle{N}\rangle{|}_{f_{s}}}<\omega_{c},
\label{NoNetPopulationInversion}
\end{equation}
where $\omega_{c}$ is the centroid of the spectrum given by
\begin{equation}
\label{centroid}
\omega_{c}=\int_{-\infty}^{\infty} d\omega\ \omega\bar{A}(\omega).
\end{equation}
The condition in Eq.\ (\ref{NoNetPopulationInversion}) implies the following: Given some nonequilibrium distribution function $f_{s}$, one
can have a positive temperature solution if and only if the average energy per particle is smaller than the centroid of the spectrum. In other words,
a positive temperature solution exists if and only if there is no net population inversion. 
Similarly, the corollary \ref{existence:Corollary} states that there exists a negative temperature solution for a system exhibiting population inversion:
\begin{equation}
\frac{\langle{E}\rangle{|}_{f_{s}}}{\langle{N}\rangle{|}_{f_{s}}}>\omega_{c}.
\end{equation}

The 
advantage of the broadband limit is that one may write the measurement conditions, as well as the condition for the existence of a
solution, in terms of the local expectation values of the energy and occupancy directly, instead of using the rate of particle and energy flow
into the probe. We also do not need to introduce a ``characteristic tunneling rate.''
We note that $\omega_{c}$ in Eq.\ (\ref{centroid}) is the centroid since the local
spectrum $\bar{A}$ normalizes to unity within the broadband limit (see Appendix \ref{sec:AppC}).

A local measurement by a weakly-coupled 
broadband thermoelectric probe is {\em ideal} in the sense 
that the result is independent of the properties of the probe, and depends only on the nonequilibrium
state of the system and the subsystem thereof sampled by the probe.
Such a measurement provides more than just an operational definition of the local temperature and voltage
of a nonequilibrium quantum system, 
since the thermodynamic variables are determined directly by the moments (\ref{eq:moments}) of the local (nonequilibrium) energy distribution.


\subsection{Nonunique Measurements: The Narrowband Limit}
\label{sec:Narrowband}
A narrowband probe is one that samples the system only within a very narrow window of energy. The extreme case of such a probe-system coupling
would be a Dirac-delta function:
\begin{equation}
\Gamma_{p}(\omega)=2\pi V_{p}^{\dagger}V_{p}\delta(\omega-\omega_{0}),
\end{equation}
which gives $\myT_{ps}(\omega)=2\pi\Tr{ V_{p}A(\omega)V_{p}^{\dagger}}\delta(\omega-\omega_{0})$ which we write simply as
\begin{equation}
\label{NBTransmission}
\myT_{ps}(\omega)=\gamma(\omega)\ \delta{(\omega-\omega_{0})},
\end{equation}
where $\gamma(\omega)=2\pi\Tr{ V_{p}A(\omega)V_{p}^{\dagger}}$ has dimensions of energy.

We previously noted that Theorem \ref{ThmOnsager} does not hold for $\myT_{ps}$ given by Eq.\ (\ref{NBTransmission}).
One can verify straightforwardly that, for a probe-sample transmission that is extremely narrow, we will have
\begin{equation}
\Lfun{0}{ps}\Lfun{2}{ps}-\big(\Lfun{1}{ps}\big)^{2}=0.
\end{equation}
This results in a nonunique solution since following the proof of theorem \ref{ThmUniqueness} would
give us [cf.\ Eq.\ (\ref{eqn:Onsager})] $\Delta{I^{(1)}_{p}}=0$.
In fact, it would lead to a family of solutions. 

We may solve for the solution explicitly. The number current reduces to
\begin{equation}
I^{(0)}_{p}=\frac{\gamma(\omega_{0})}{h}\big(f_{p}(\omega_{0})-f_{s}(\omega_{0})\big),
\end{equation}
while the heat current is given by
\begin{equation}
I^{(1)}_{p}=(\omega_{0}-\mu_{p})\frac{\gamma(\omega_{0})}{h}\big(f_{p}(\omega_{0})-f_{s}(\omega_{0})\big),
\end{equation}
which trivially vanishes for vanishing number current. Therefore, the
family of solutions to the measurement is simply given by
\begin{equation}
f_{p}(\omega_{0};\mu_{p},T_{p})=f_{s}(\omega_{0}),
\end{equation}
which is linear in the $\mu_{p}-T_{p}$ plane and is given by
\begin{equation}
\mu_{p}= \omega_{0} - T_{p}\log\bigg(\frac{1-f_{s}(\omega_{0})}{f_{s}(\omega_{0})}\bigg).
\end{equation}
$f_{s}(\omega)$ has the following explicit form:
\begin{equation}
f_{s}(\omega)= \frac{\Tr{V_p G^{<}(\omega)V_p^\dagger}}{2{\pi}i\Tr{V_p A(\omega)V_p^\dagger}}.
\end{equation}

A {\em narrowband probe} is therefore unsuitable for thermoelectric measurements. Even if a probe were to sample 
the system at just two distinct energies $\omega_{1}$ and $\omega_{2}$, theorem \ref{ThmOnsager} would hold and the thermoelectric 
measurement would be unique.
Indeed, the narrowband probe is a pathological case whose only function is to highlight a certain theoretical limitation for the
measurement of the temperature and voltage. 

\subsection{Example: Two-level system}
\label{sec:TwoLevelSystem}
Net population inversion is essentially a quantum phenomenon, since classical Hamiltonians are generally unbounded above
due to the kinetic energy term, i.e., there does
not exist a finite $c\in\mathbb{R}$ that satisfies $\langle H\rangle<c$. In other words, $\omega_{c}\to\infty$ generally holds for classical systems
and negative temperatures are not possible. The simplest quantum system where a net population inversion can be achieved is a two-level system.
We therefore illustrated our results for a two-level system in \cref{fig:Lemma1,fig:Lemma2,fig:Theorem2,fig:Corollary3}.

\begin{figure}
\captionsetup{justification=raggedright,
singlelinecheck=false}
  \centering
  \begin{tabular}{@{}c@{}}
    \includegraphics[width=2.5in]{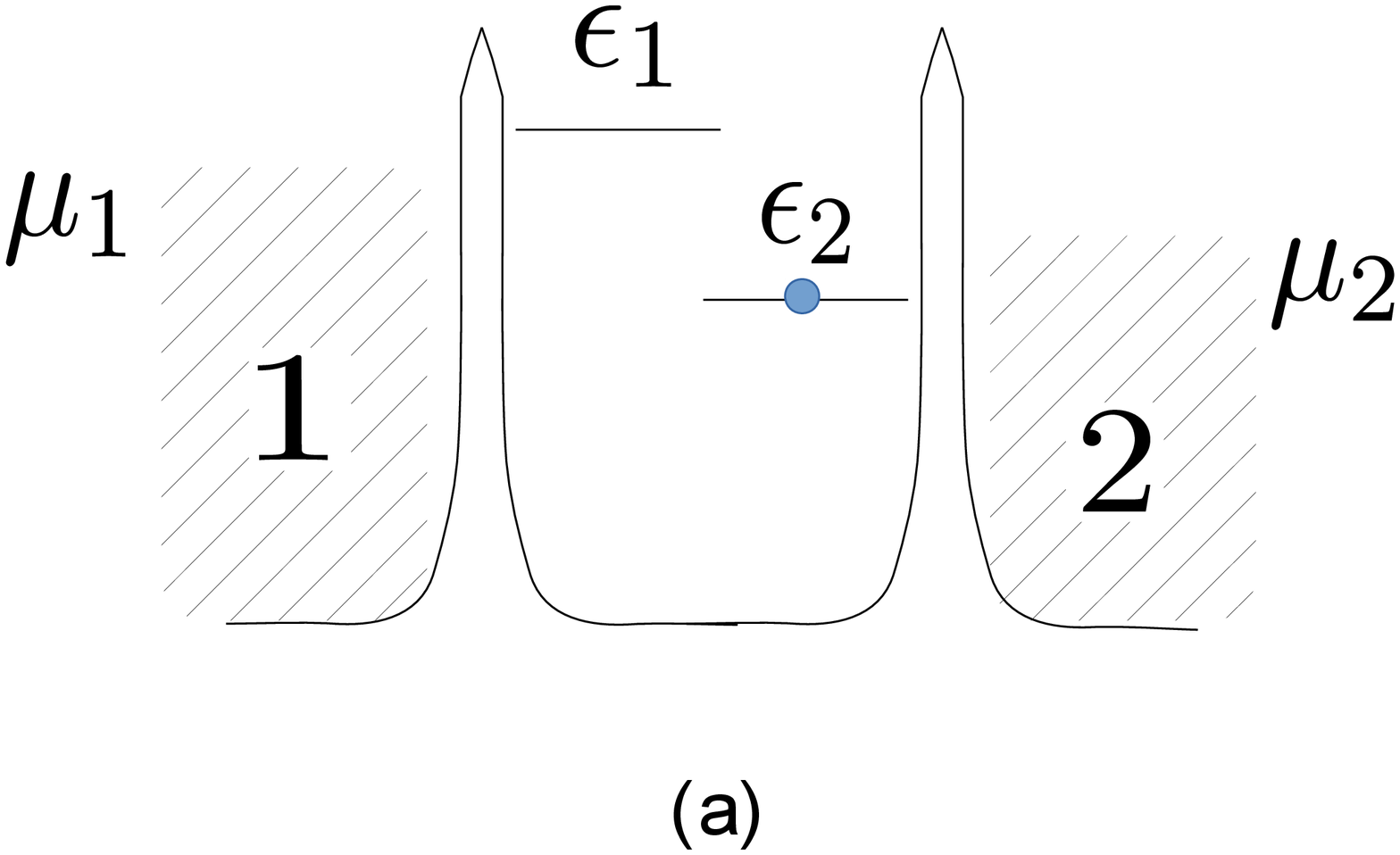} 
  \end{tabular} 
  \begin{tabular}{@{}c@{}}
    \includegraphics[width=2.5in]{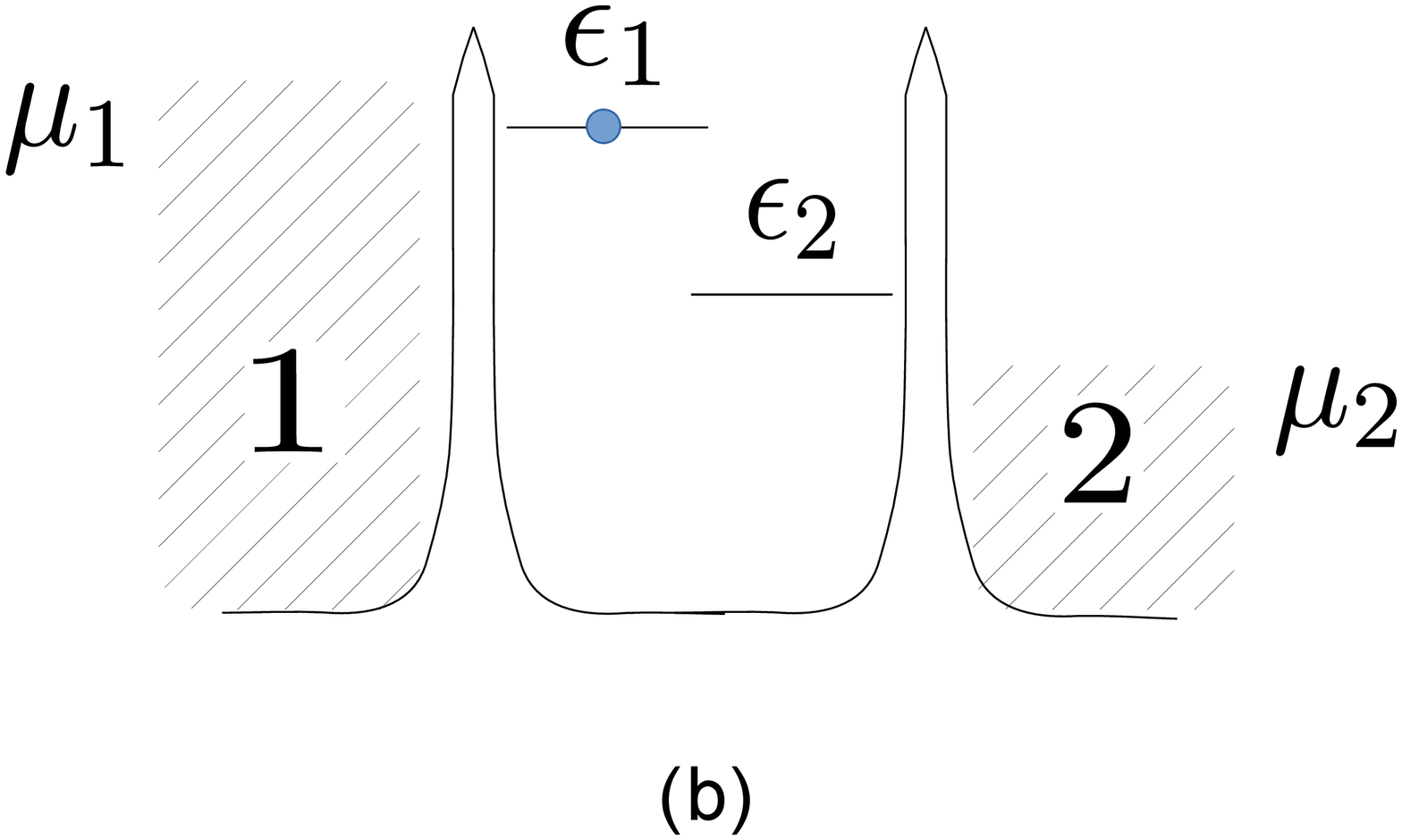} 
  \end{tabular} 
  \caption{Schematic diagram of a two-level system coupled to two electron reservoirs under bias.
   (a) Bias condition not leading to population inversion.
   (b) Bias condition leading to population inversion due to direct injection into excited state.
          }
\label{fig:myfig}
\end{figure}

The system Hamiltonian here was taken to be
\begin{equation}
    H=     \begin{bmatrix}
    \epsilon_{1} & V  \\
    V^{*} & \epsilon_{2}
  \end{bmatrix}, 
\end{equation}
whose values were set as $V=\frac{2(1-i)}{3}$, $\epsilon_{1}=1/3$ and $\epsilon_{2}=-1/3$, such that the
eigenvalues are $\epsilon_{\pm}=\pm 1$ and units are taken as $eV$. 
We introduce two reserviors that
are strongly coupled locally to each site with $\Gamma_{1}=\text{diag}(0.5,0)$ and $\Gamma_{2}=\text{diag}(0,0.5)$, while the probe coupling is taken as 
$\Gamma_{p}=\text{diag}(0.01,0.1)$, which is about five times weaker than the coupling to the reservoirs that bias the system. 

We used two different bias conditions: (a) To illustrate the case without a net population inversion in \cref{fig:Lemma1,fig:Lemma2,fig:Theorem2},
the reservoirs had a symmetric ($\mu_{1}+\mu_{2}=0$) voltage bias $\mu_{1}-\mu_{2}=1eV$;
(b) to illustrate the case with a net population inversion in \cref{fig:Corollary3}, the reservoirs had a symmetric voltage bias of 
$\mu_{1}-\mu_{2}=4eV$. The two reservoirs are held at $T=300K$ for both cases.

It has been previously noted that
the probe-system coupling strength does not strongly affect 
the measured temperature and voltage even when varied over several orders of magnitude \cite{Meair14}, but
we remind the reader that our theoretical results depend upon the assumption of a weakly-coupled probe (noninvasive measurements).
How {\em weak} is {\em weak enough} is a different, and perhaps more subtle,
theoretical question. Numerically, however, we do find that the probe measurements are not much altered even when the probe coupling
strength is comparable
to that of the strongly-coupled reservoirs.

\section{Conclusions}
\label{sec:conclusions}

The local temperature and voltage of a nonequilibrium quantum system are defined in terms of the equilibration 
of a noninvasive thermoelectric probe,
locally coupled to the system.  The simultaneous temperature and voltage measurement is shown to be unique for any system of fermions in steady state,
arbitrarily far from equilibrium, with arbitrary interactions within the system, and the conditions for the existence of
a solution are derived.  In particular, it is shown that a positive temperature solution exists provided the system does not have a net local population
inversion; in the case of population inversion, a unique negative temperature solution is shown to exist.  These results provide a firm
mathematical foundation for temperature and voltage measurements in quantum systems far from equilibrium.

Our analysis reveals that a simultaneous temperature and voltage measurement is uniquely determined by the local spectrum 
and nonequilibrium distribution 
of the system [cf.\ Eq.\ (\ref{NewEquilibrium})], and is independent of the properties of the probe for broadband coupling (ideal probe).
Such a measurement therefore provides a {\em fundamental definition} of local temperature and voltage, which is not merely operational.

In contrast, prior theoretical work relied almost exclusively on operational definitions
\cite{Engquist81,Dubi2009b,Jacquet2009,Dubi2009,Caso2011,Jacquet2012,Sanchez2011,Caso10,Bergfield2013demon,Bergfield2015,Ye2015}, 
leading to a competing panoply of often contradictory predictions for the measurement 
of such basic observables as temperature and voltage.  Measurements of temperature or voltage, taken separately 
(see, e.g., Refs.\ \onlinecite{Engquist81,Jacquet2012}), 
are shown to be ill-posed:  a thermometer out of electrical equilibrium with a system produces an error due to the Peltier effect across the
probe-sample junction, while a potentiometer out of thermal equilibrium with a system produces an error due to the Seebeck effect.

Our results put the local thermodynamic variables temperature and voltage
on a mathematically rigorous footing for fermion systems under very general nonequilibrium steady-state conditions,
a necessary first step toward the construction of
nonequilibrium thermodynamics \cite{Ruelle2000,Casas-Vazques2003,Lebon2008,Cugliandolo2011,Jacquet2012,Stafford2014,Esposito2015,Shastry2015}.
Our analysis includes the effect of interactions with bosonic degrees of freedom (e.g., photons, phonons, etc.) on the fermions.  However, the 
temperatures of the bosons themselves \cite{Ming10,Galperin2007} were not addressed in the present analysis. 
Moreover, we did not explicitly consider magnetic systems, which require separate consideration of the spin degree of freedom, and its polarization.
Future investigation of probes that exchange bosonic or spin excitations may enable similarly rigorous analysis of local thermodynamic variables in
bosonic and magnetic systems, respectively.

\section{Acknowledgements}

A.S would like to thank Janek Wehr for interesting and useful lectures in mathematical physics. This work
was supported by the U.S. Department of Energy, Office of
Science under Award No. DE-SC0006699. 

\appendix

\section{The nonequilibrium steady state}
\label{sec:AppA}

The nonequilibrium steady state is described by a density matrix $\hat{\rho}$ that is time-independent.
The expectation values of observables are given by their usual prescription in statistical physics, e.g.,
\begin{equation}
\langle\hat{Q}\rangle=\Tr{\hat{\rho}\hat{Q}}= \sum_{\mu,\nu}{\rho_{\mu\nu}\langle\nu|\hat{Q}|\mu\rangle}.
\end{equation}

The ``lesser" and ``greater" Green's functions \cite{Stefanucci2013} used in the paper are defined as follows
\begin{equation}
G^{<}_{\alpha\beta}(t)\equiv i\langle {d^{\dagger}_{\beta}(0)}d_{\alpha}(t)\rangle,
\end{equation}
while its Hermitian conjugate is
\begin{equation}
G^{>}_{\alpha\beta}(t)\equiv -i\langle d_{\alpha}(t){d^{\dagger}_{\beta}(0)}\rangle,
\end{equation}
where
\begin{equation}
d_{\alpha}(t)= e^{i\frac{\hat{H}}{\hbar}t}d_{\alpha}(0)e^{-i\frac{\hat{H}}{\hbar}t}
\end{equation}
evolves according to the Heisenberg equation of motion for a system with Hamiltonian $\hat{H}$.
Here, $\alpha$, $\beta$ denote basis states in the 1-body Hilbert space of the system.

The spectral representation uses the eigenbasis of the Hamiltonian $\hat{H}|\nu\rangle=E_{\nu}|\nu\rangle$, where $\nu$ denotes a many-body energy
eigenstate.
One may write the ``lesser" Green's function as
\begin{equation}
\begin{aligned}
G^{<}_{\alpha\beta}(\omega)=2\pi i\sum_{\mu,\mu{'},\nu}\rho_{\mu\nu}&\langle\nu|d^{\dagger}_{\beta}|\mu{'} \rangle\langle\mu{'} |{d_{\alpha}}|\mu\rangle\\
&\times\delta\bigg(\omega-\frac{E_{\mu}-E_{\mu{'}}}{\hbar}\bigg),
\end{aligned}
\end{equation}
while the ``greater" Green's function becomes
\begin{equation}
\begin{aligned}
G^{>}_{\alpha\beta}(\omega)=-2\pi i\sum_{\mu,\mu{'},\nu}\rho_{\mu\nu}&\langle\nu|{d_{\alpha}}|\mu{'} \rangle\langle\mu{'} |d^{\dagger}_{\beta}|\mu\rangle\\
&\times\delta\bigg(\omega-\frac{E_{\mu{'}}-E_{\nu}}{\hbar}\bigg).
\end{aligned}
\end{equation}
The spectral function $A(\omega)$ is given by
\begin{equation}
A(\omega)\equiv \frac{1}{2\pi i}\bigg(G^{<}(\omega)-G^{>}(\omega)\bigg),
\end{equation}
and can be expressed in the spectral representation as
\begin{equation}
\label{SpectralFunction}
\begin{aligned}
A_{\alpha\beta}(\omega)=\sum_{\mu,\mu{'},\nu}\bigg[\rho_{\mu\nu}&\langle\nu|d^{\dagger}_{\beta}|\mu{'} \rangle\langle\mu{'} |{d_{\alpha}}|\mu\rangle\\
 +&  \rho_{\nu\mu{'}}\langle\mu{'}|{d_{\alpha}}|\mu \rangle\langle\mu |d^{\dagger}_{\beta}|\nu\rangle\bigg]\\
 & \times\delta\bigg(\omega-\frac{E_{\mu}-E_{\mu{'}}}{\hbar}\bigg).
\end{aligned}
\end{equation}

\subsection{Sum rule for the spectral function}
\label{sec:AppC} 
Eq.\ (\ref{SpectralFunction}) leads to the following sum rule for the spectral function:
\begin{equation}
\label{SumRule}
\begin{aligned}
\int_{-\infty}^{\infty}d\omega A_{\alpha\beta}(\omega)=&\sum_{\mu,\nu}\rho_{\mu\nu}\langle\nu|d^{\dagger}_{\beta}d_{\alpha}|\mu\rangle\\
&\ \ \ \ \ \ \ \ \ \ \ \ \ \ +\sum_{\mu{'},\nu}\rho_{\nu\mu{'}}\langle\mu{'}|d_{\alpha}d^{\dagger}_{\beta}|\nu\rangle\\
=&\sum_{\mu,\nu}\rho_{\mu\nu}\langle\nu|d^{\dagger}_{\beta}d_{\alpha}+d_{\alpha}d^{\dagger}_{\beta}|\mu\rangle\\
=&\sum_{\mu,\nu}\rho_{\mu\nu}\langle\nu|\delta_{\alpha\beta}|\mu\rangle\\
=&\sum_{\mu,\nu}\rho_{\mu\nu}\delta_{\mu\nu}\delta_{\alpha\beta}\\
=&\delta_{\alpha\beta}\Tr{\hat{\rho}}\\
=&\delta_{\alpha\beta}.
\end{aligned}
\end{equation}

In our theory of local thermodynamic measurements, the quantity of interest is the local spectrum of the system sampled by the probe $\bar{A}(\omega)$,
defined in Eq.\ (\ref{LocalSpectrum}). This obeys a further sum rule in the
broadband limit ({\em ideal probe}), discussed below.

\subsubsection{Local spectrum in the broadband limit}

The probe-system coupling is energy independent in the broadband limit, $\Gamma^{p}(\omega)=\text{const}$, and we write $\Tr{\Gamma^{p}}=\bar{\Gamma}^{p}$
for its trace.
The local spectrum sampled by the probe $\bar{A}(\omega)$ defined in Eq.\ (\ref{LocalSpectrum})
can be written in the broadband limit as
\begin{equation}
\bar{A}(\omega)=\frac{1}{\bar{\Gamma}^{p}}\sum_{\alpha,\beta}\langle\beta|\Gamma^{p}|\alpha\rangle A_{\alpha\beta}(\omega).
\end{equation}
In this limit, it obeys a further sum rule:
\begin{equation}
\begin{aligned}
\int_{-\infty}^{\infty}d\omega \bar{A}(\omega)&= \frac{1}{\bar{\Gamma}^{p}}\sum_{\alpha,\beta}\langle\beta|\Gamma^{p}|\alpha\rangle \int_{-\infty}^{\infty}d\omega A_{\alpha\beta}(\omega)\\
&=\frac{1}{\bar{\Gamma}^{p}}\sum_{\alpha,\beta}\langle\beta|\Gamma^{p}|\alpha\rangle\delta_{\alpha\beta}\\
&=1.
\end{aligned}
\label{eq:sumrulelocal}
\end{equation}
The broadband limit is special in that the measurement is determined by the local properties of the system itself, and is not influenced by the spectrum of
the probe.  In this limit, the local spectrum $\bar{A}(\omega)$ obeys the sum rule (\ref{eq:sumrulelocal}) since the probe samples the same
subsystem at all energies.  One should not expect such a local sum rule to hold outside the broadband limit, since the probe samples different subsystems
at different energies.

\subsection{Diagonality of $\hat{\rho}$}

We have, for any observable $\hat{Q}$,
\begin{equation}
\begin{aligned}
\langle\hat{Q}(t)\rangle&= \sum_{\mu,\nu} \rho_{\mu\nu}\langle\nu|\hat{Q}(t)|\mu\rangle\\
&=\sum_{\mu,\nu} \rho_{\mu\nu}\langle\nu|e^{i\frac{\hat{H}}{\hbar}t}\hat{Q}e^{-i\frac{\hat{H}}{\hbar}t}|\mu\rangle\\
&=\sum_{\mu,\nu} \rho_{\mu\nu}e^{-i\frac{E_{\mu}-E_{\nu}}{\hbar}t}\langle\nu|\hat{Q}|\mu\rangle.
\end{aligned}
\end{equation}

The system observables must be independent of time in steady state. Therefore $\hat{\rho}$ must be diagonal in the energy basis,
as seen from the above equation. The nondiagonal parts
of $\hat{\rho}$ in the energy basis, when they exist, must be in a degenerate subspace so that $E_{\mu}=E_{\nu}$ in the above equation.

For states degenerate in energy, the boundary conditions determining the nonequilibrium steady state will determine the basis in which $\hat{\rho}$
is diagonal. Henceforth, we work in that basis.

\subsection{Positivity of $-iG^{<}(\omega)$ and $iG^{>}(\omega)$}

Working in the energy eigenbasis in which $\hat{\rho}$ is diagonal, 
\begin{eqnarray}
\lefteqn{
-i\langle\alpha|G^{<}(\omega)|\alpha\rangle
 \equiv 
-iG^{<}_{\alpha\alpha}(\omega)  
=}
\nonumber \\
 & & \!\!\!\!\!\!\! 2\pi\sum_{\mu,\mu{'}}\rho_{\mu\mu}\abs{\langle\mu|d^{\dagger}_{\alpha}|\mu{'}\rangle}^{2}
\delta\bigg(\omega-\frac{E_{\mu}-E_{\mu{'}}}{\hbar}\bigg) 
\geq 0.
\end{eqnarray}
Similarly,
\begin{eqnarray}
\lefteqn{i\langle\alpha|G^{>}(\omega)|\alpha\rangle \equiv
iG^{>}_{\alpha\alpha}(\omega) =} 
\nonumber \\
& & \!\!\!\!\!\!\! 2\pi\sum_{\mu,\mu{'}}\rho_{\mu\mu}\abs{\langle\mu|d^{\dagger}_{\alpha}|\mu{'}\rangle}^{2}
\delta\bigg(\omega-\frac{E_{\mu{'}}-E_{\mu}}{\hbar}\bigg)
\geq 0.
\end{eqnarray}
It follows that
\begin{equation}
\langle\alpha|A(\omega)|\alpha\rangle=
\frac{1}{2\pi}\langle\alpha|-\!iG^{<}(\omega)+iG^{>}(\omega)|\alpha\rangle
\geq 0.
\end{equation}
Therefore, all three operators $-iG^{<}(\omega)$, $iG^{>}(\omega)$, and $A(\omega)$ are positive-semidefinite.

\subsection{$0\leq f_{s}(\omega)\leq1$}
\label{Property:nonequilibriumdistribution}

The nonequilibrium distribution function $f_{s}(\omega)$ was defined in Eq.\ (\ref{nonequilibriumdistribution}) as
\begin{equation}
f_{s}(\omega)\equiv\frac{\Tr{\Gamma^{p}(\omega)G^{<}(\omega)}}{2{\pi}i\Tr{\Gamma^{p}(\omega)A(\omega)}}.
\label{eq:fs}
\end{equation}
We have $\Gamma^{p}(\omega)> 0$ by causality \cite{Stefanucci2013}:
\begin{equation}
\text{Im}\ \Sigma^{r}_{p}(\omega)=-\frac{1}{2}\Gamma^{p}(\omega)<0.
\end{equation}
Let $\Gamma^{p}|\gamma_{p}\rangle=\gamma_{p}|\gamma_{p}\rangle$, where $\gamma_{p}\geq0$ and some $\gamma_{p}$ satisfy $\gamma_{p}>0$.
The energy dependence is taken to be implicit.
The traces in Eq.\ (\ref{eq:fs}) may be evaluated in the eigenbasis of $\Gamma^p$, yielding:
\begin{equation}
\begin{aligned}
f_{s}(\omega)&= \frac{\sum_{\gamma_{p}}\gamma_{p}\langle\gamma_{p}|G^{<}(\omega)|\gamma_{p}\rangle}{2\pi i\sum_{\gamma_{p}}\gamma_{p}\langle\gamma_{p}|A(\omega)|\gamma_{p}\rangle}\\
&= \frac{\sum_{\gamma_{p}}\gamma_{p}\langle\gamma_{p}|-i G^{<}(\omega)|\gamma_{p}\rangle}{\sum_{\gamma_{p}}\gamma_{p}\langle\gamma_{p}|-iG^{<}(\omega)+iG^{>}(\omega)|\gamma_{p}\rangle}.
\end{aligned}
\end{equation}
Therefore
\begin{equation}
0\leq f_{s}(\omega)\leq1.
\end{equation}

\bibliography{./refs}
\end{document}